\documentclass[12pt,a4paper,draftcls,onecolumn]{IEEEtran}
\date{}
\usepackage{amsfonts}      
\usepackage{amsmath}
\usepackage{mathrsfs}
\usepackage{amsfonts}      
\usepackage{amssymb}

\usepackage{graphics} 
\usepackage{epsfig}
\usepackage{pmat}
\usepackage{dsfont}

\usepackage{multirow}

\usepackage{makecell}
\usepackage{booktabs}
\usepackage{setspace}

\usepackage{color}
\usepackage{enumerate,cite}

\newtheorem{theorem}{Theorem}[section]
\newtheorem{remark}[theorem]{Remark}
\newtheorem{lemma}[theorem]{Lemma}
\newtheorem{proposition}[theorem]{Proposition}
\newtheorem{definition}[theorem]{Definition}
\newtheorem{example}[theorem]{Example}

\usepackage{algorithm}
\usepackage{algorithmic}

\title{Finite-Frequency Model Order Reduction of Linear Systems via Parameterized Frequency-dependent Balanced Truncation
\thanks{This work was supported by NSFC under Grant (61304143) and the High-End Foreign Expert Program of the P.~R.~China (GDT20153100033). }}

\author{Xin~Du$^{1,2}$~~
Peter~Benner$^{1,2}$
\thanks{$^1$ Max Planck Institute for
Dynamics of Complex Technical Systems, Sandtorstra\ss e 1, 39106
Magdeburg, Germany. \texttt{benner@mpi-magdeburg.mpg.de}} \\ 
\thanks{$^2$ School of Mechatronic Engineering and Automation, Shanghai University,  Shanghai, 200072, P.~R.~China. \texttt{duxin@shu.edu.cn}}
}

\begin{document}
\centerfigcaptionstrue

\maketitle 

\begin{abstract}
Balanced truncation is the most commonly used model order reduction scheme in control engineering.  This is due to its favorable properties of automatic stability preservation and the existence of a computable error bound, enabling the adaption of the reduced model order to a specified tolerance.  It aims at minimizing the worst case error of the frequency response over the full infinite frequency range. If a good approximation only over a finite frequency range is required, frequency-weighted or frequency-limited balanced truncation variants can be employed.
In this paper, we study this finite-frequency model order reduction (FF-MOR) problem for linear time-invariant (LTI) continuous-time systems within the framework of balanced truncation. Firstly, we construct a family of parameterized frequency-dependent (PFD) mappings which transform the given LTI system to either a discrete-time or continuous-time PFD system. The relationships between the maximum singular value of the given LTI system over pre-specified frequency ranges and the maximum singular value of the PFD mapped systems over the entire frequency range are established. By exploiting the properties of the discrete-time PFD mapped systems, a new parameterized frequency-dependent balanced truncation (PFDBT) method providing a finite-frequency type error bound with respect to the maximum singular value of the error systems is developed. Examples are included for illustration.
\end{abstract}




%

\section{Introduction}

\subsection{Problem Formulation}
Model order reduction (MOR) is an ubiquitous tool in the analysis and simulation of dynamical systems, control design, circuit simulation, structural dynamics, computational fluid dynamics, and many more areas in the computational sciences and engineering; see, e.g., \cite{morAnt05,morBenMS05,morBenHM11,morSchVR08}. Modeling of complex physical processes often leads to dynamical systems with high-dimensional state-spaces, so that the corresponding system is of large order $n$.  This may lead to difficulties in the simulation, optimization, control and design of such systems due to memory restrictions and (run) time limitations for the execution of the related algorithms. In general, the purpose of MOR is to produce a lower dimensional system that has similar response characteristics as the original system with far lower storage requirements and largely reduced evaluation time. In this paper, we focus on the MOR problem for linear time-invariant (LTI) dynamical systems:
\begin{equation}\label{originalsystem}
G:\quad\left\{\begin{array}{l}
{\dot x(t) = Ax(t) + Bu(t)}  \\
   {y(t) = Cx(t) + Du(t)}  \\
   \end{array}\right.
  ~\Leftrightarrow~ G(\jmath\omega ):=  C{(\jmath\omega I - A)^{ - 1}}B + D,
\end{equation}
\noindent
where $A\in \mathbb R^{n\times n},  B \in \mathbb
R^{n\times m},  C \in \mathbb R^{p\times n},  D \in \mathbb R^{p\times
m}$, $x(t) \in \mathbb R^n$ is the state vector, $u(t)\in \mathbb
R^m$ is the input signal, $y(t)\in \mathbb C^p$ is the output
signal. The imaginary unit is denoted by $\jmath$, and $\omega\in\mathbb R$ is related to the operating frequency $f$ (measured in Hertz) of the LTI system via $\omega=2\pi f$ .
By abuse of notation, we denote the LTI system as well as its transfer function by $G$.
A realization of the LTI system (\ref{originalsystem}) is given by the matrix tuple ($A, B, C, D$).
When appropriate, we will also use the equivalent notation $ \pmatset{1}{0.36pt}
  \pmatset{0}{0.2pt}
  \pmatset{2}{4pt}
  \pmatset{3}{4pt}
  \pmatset{4}{4pt}
  \pmatset{5}{4pt}
  \pmatset{6}{4pt}    \begin{pmat}[{|}]
       A   &   B \cr\-
       C   &   D \cr
      \end{pmat}$, which is common in control theory.

The aim of MOR then is to approximate the LTI system (\ref{originalsystem}) by a reduced-order LTI system:
\begin{equation} \label{reducedmodel}
G_r\quad:\left\{\begin{array}{l}
{{\dot x}_r(t) = A_rx_r(t) + B_ru(t)}  \\
   {y_r(t) = C_rx_r(t) + D_ru(t)}  \\
   \end{array}\right.
  ~\Leftrightarrow~ G_r(\jmath\omega ):=  C_r{(\jmath\omega I - A_r)^{ - 1}}B_r + D_r,
\end{equation}
\noindent
where $A_r\in \mathbb R^{r\times r}, B_r \in \mathbb
R^{n\times m}, C_r \in \mathbb R^{p\times n}, D_r \in \mathbb R^{p\times
m}$ with $r \ll n$, and so that $y(t)\approx y_r(t)$ for $t$ in some chosen time range and for all admissible input functions $u(t)$.
In other words, in order to replace the original model successfully, the reduced-order model should approximate the input-output behavior of the original system as well as possible.  This underlying requirement on the reduced-order model means that the MOR problem inherently depends on the chosen class of input signals, that is, different types of input signals will lead to different MOR problems with respect to the approximation performance.
From the frequency-domain viewpoint, signals can be classified into entire-frequency (EF) type signals and finite-frequency (FF) type signals, as listed in Table~\ref{t9}; cf.~\cite{GKYP_Iwasaki1}.
\begin{table}[b]\small
\setlength{\abovecaptionskip}{0pt}
\setlength{\belowcaptionskip}{-2pt}
\centering
		\caption{Different frequency ranges for input signals}
\renewcommand{\arraystretch}{1.2}
\begin{tabular}{|c|c|c|c|}
\hline
\multicolumn{1}{|c|} {\textbf{EF} } & \multicolumn{3}{c|}{{\kern 30pt}\textbf{FF} (finite-frequency)} \\
\cline{2-4}
 (entire-frequency)                             & \textbf{LF} (low-frequency)           & \textbf{MF} (middle-frequency)          & \textbf{HF} (high-frequency)                       \\ \hline

$\omega \in \Omega: (-\infty, +\infty)  $                      & $\omega \in \Omega_l: [-\varpi_l, +\varpi_l]$    & $\omega \in \Omega_m: [\varpi_1, \varpi_2]$    & $\omega \in \Omega_h: (-\infty, -\varpi_h] \cup [\varpi_h, +\infty) $                                   \\ \hline
\end{tabular}
\label{t9}
\end{table}

\noindent
Obviously, such a classification of the frequency range of input signals will give rise to several classes of MOR problems:  EF-MOR when considering the full frequency range, and FF-MOR (including LF-MOR, MF-MOR, and HF-MOR) for limited frequency ranges, respectively. In case that there exists no \emph{a priori} known frequency information of the input signals or the frequency of input signals belongs to a very wide range, EF-MOR problems will be the appropriate choice, and a uniform approximation performance over the entire frequency range should be taken into consideration.  For many practical cases, though, a certain range for the frequency of the input signals is pre-known.  In these situations, it will be better to resort to a FF-MOR formulation since only the in-band input-output behavior of the original system is needed to be captured; cf., e.g., \cite{morGhaS08,morObiA01,morMus00}.
Thus, good in-band approximation performance can be expected while neglecting the out-of-band approximation performance, or, in other words, a better approximation quality in-band at the same reduced order is to be expected than for methods trying to approximate uniformly in the entire frequency band.

\subsection{Literature Review}
During the last decades, many efficient approaches such as balanced truncation \cite{morMoo81,morAlF87}, moment matching \cite{morBai02,morFre03}, and modal truncation \cite{morBesetal13} have been developed from different fields; see also the books \cite{morAnt05,morBenMS05,morBenHM11,morSchVR08} and the recent survey \cite{morBauBF14}.
Among them, balanced truncation stands out for its beneficial properties relevant in control design, i.e., stability preservation and computable error bound,
allowing for an automatic reduced-order model generation.  Here, we focus on balanced truncation, and therefore in the following mainly review the literature with regard to attempts of adopting balanced truncation to the FF-MOR framework.

The idea underlying balanced truncation consists in transforming the state space system into a balanced form whose controllability and observability Gramians become diagonal and equal, together with a truncation of those states that are both difficult to reach and to observe. The standard version of balanced truncation is often called Lyapunov balancing (LyaBT), see, e.g., \cite{morGugA04}, and was first introduced by Moore in 1981 \cite{morMoo81}.
The reduced-order model obtained by LyaBT has diminishing error for increasing frequencies, but takes the maximum error often at $\omega=0$. In order to match the DC gain, i.e., to have zero error at $\omega=0$, but allowing a larger error at large frequencies, Liu and Anderson developed the singular perturbation approximation (SPA) scheme \cite{morLiuA86}, which is also based on a balanced realization of the LTI system. Both, LyaBT and SPA, are widely appreciated and recognized as the most suitable techniques for EF-MOR problems since both of them provide a computable \emph{a priori} simple error bound, called EF-type error bound in the following, with respect to the following entire-frequency approximation performance index:
\begin{equation}\label{EFindex}
\begin{array}{l}
\sigma_{\max}\left(G(\jmath \omega ) - {G_r}(\jmath \omega)\right)  , {\kern 6pt} \omega \in \Omega= (-\infty,+\infty)
\end{array}.
\end{equation}
Though this performance index related to the $H_\infty$-norm of the error system, is not minimized by LyaBT and SPA, the computed reduced-order models usually get close to optimal \cite{morAnt05,morMinBF14}.
The error bound makes it possible to choose the reduced order $r$ automatically.
As mentioned above, LyaBT generally leads to good high-frequency approximation performance since the reduced-order models generated via LyaBT matches the original model exactly at $|\omega|=\infty$, while SPA generally leads to good low-frequency approximation performance as the corresponding reduced-order models match the original model exactly at $\omega=0$. However, it is unclear how good the in-band approximation performance over a specified HF (LF) range is, since only the EF-type error bound is known for LyaBT and SPA.

In order to make the standard LyaBT scheme more suitable for solving FF-MOR problems, several modified BT schemes have been developed. Frequency-weighted balanced truncation (FWBT) and frequency-limited Gramians balanced truncation (FGBT) are two popular ones for this purpose and were studied during the last 25 years. The common procedure of FWBT is to build a frequency-weighted model first by introducing input/output frequency weighted transfer functions and then apply the standard LyaBT or SPA procedure on the weighted model; see, e.g., \cite{morGhaS08,morZho95,morSreA89,morWanSL89,morSreetal13,morHouS09,morEnn84}.
Indeed, good frequency-specific approximation performance may be obtained if the selected weighting function is appropriately chosen. However, the design iterations to search for such a weighting transfer function can be tedious and time consuming. Besides, FWBT also suffers from the drawback of the increased order of the weighted plant model.

FGBT was first introduced by Gawronski and Juang in \cite{morGawJ90}. This methodology stems from the consideration of extending the definition of standard Gramians to the frequency-limited case and then applying the standard balanced truncation procedures to the frequency-limited Gramians \cite{morZadZ05,morGhaS08b,morShaT14}. An implementation of this method for truly large-scale systems was recently suggested in \cite{morBenKS15}.
As has been pointed out in \cite{morGugA04,morDuY10}, FGBT may be invalid in some cases as the solutions of the ``frequency-limited Lyapunov equations" cannot be guaranteed to be positive semi-definite, and it provides no error bound.  Although there exist several modified FGBT schemes, see, e.g., \cite{morGugA04,morGhaS08b} to overcome those drawbacks, good in-band performance generally cannot be guaranteed. More importantly, both FWBT and FGBT continue to use the EF-type index (\ref{EFindex}) to evaluate the actually desired \emph{finite-frequency} approximation performance. This incompatibility between the intrinsic requirement and the achievement of the method yields many deficiencies. Since only EF-type error bounds are available, whether or not the in-band approximation performance has been improved cannot be guaranteed.  In particular, FWBT and FGBT may give rise to poor in-band approximation performance together with a large error bound in some cases.

In \cite{MPIMD13-02}, we studied the FF-MOR problem from the perspective of achieving good approximation quality locally by devising a balanced truncation style method satisfying an error bound at a prescribed frequency. The method shows good approximation quality locally in a neighborhood of the given frequency point, and this neighborhood is usually larger than for interpolatory (or moment-matching) methods that have zero error at the prescribed frequency.  Nevertheless, this new method does not solve the FF-MOR problem satisfactorily as it provides no error bound valid on a (half-)finite interval.

The shortcomings of the approaches to adapt balanced truncation to the FF-MOR setting motivated us to study this problem from a new FF-type error bound centered viewpoint.

\subsection{Contributions and Structure}
 In this paper, we are dedicated to solving the FF-MOR problems within the framework of balanced truncation. In contrast to existing BT schemes, we are interested in developing a new way to provide in-band error bounds by using the following FF-type  approximation performance index
\begin{equation}\label{FFindex}
\begin{array}{l}
\sigma_{\max}(G(\jmath \omega ) - {G_r}(\jmath \omega)), {\kern 6pt} \omega \in \Omega_l/ \Omega_m/ \Omega_h.
\end{array}.
\end{equation}
\noindent Compared with the EF-type index (\ref{EFindex}), adopting the FF-type index (\ref{FFindex}) is obviously more appealing for FF-MOR problems. To this end, a fundamental tool estimating the maximum singular value of an LTI system over finite-frequency ranges is developed first, and then new BT based schemes are proposed for LF-MOR problems and HF-MOR problems. In particular, the contributions of this paper are:
\begin{enumerate}
\item
By introducing an auxiliary user-defined parameter $\rho$, two kinds of discrete-time parameterized frequency-dependent (PFD) systems and two kinds of continuous-time parameterized frequency-dependent (PFD) system are constructed by a suitable mapping applied to the given continuous-time LTI system.  The mapping is determined with respect to the specified finite-frequency range. Furthermore, PFD bounded real lemmas bounding the maximum singular value of the given system over the pre-specified finite-frequency ranges are derived. It is shown that there exist special relationships between the maximum singular value of the given system over the pre-specified finite-frequency ranges and the maximum singular value of the PFD mapped systems over the entire frequency ranges.
\item
By exploiting the standard discrete-time LyaBT method and the developed PFD bounded real lemma, new PFD balanced truncation (PFDBT) schemes are proposed to solve the LF-MOR and HF-MOR problems, respectively. The new PFDBT methods generate reduced-order models and provide FF-type approximation error bounds in the sense of bounding the maximum singular value of the error system over the pre-specified frequency range.
\end{enumerate}
The remainder of this paper is organized as follows:  First, we introduce the KYP lemma and the Generalized KYP Lemma in subsection \uppercase\expandafter{\romannumeral2}.A,  and then we present the definitions of the PFD mapped systems as well as the corresponding PFD bounded real lemmas in subsections \uppercase\expandafter{\romannumeral2}.B and \uppercase\expandafter{\romannumeral2}.C. Thereafter, we present the PDFBT algorithms and the results on the FF-type error bounds for the LF-MOR and HF-MOR problems in Section \uppercase\expandafter{\romannumeral3}. Next, we demonstrate the effectiveness and advantages of the proposed results by several examples in Section \uppercase\expandafter{\romannumeral4}.  Finally, we end with a conclusion in Section \uppercase\expandafter{\romannumeral5}.

\noindent \textbf{Notation}: For a matrix $M$, $M^{T}$ and $M^{*}$ denote its transpose
and conjugate transpose, respectively. $M>0$ and $M\geq 0$ indicate a positive definite and semi-definite matrix, respectively.
The symbol $\star$ within a matrix represents symmetric entries and $He(M):=\frac12(M+M^*)$ is the Hermitian part of a matrix $M$. $\sigma_{\max}(G)$
denotes the maximum singular value of the transfer matrix $G$. $\mathfrak {Re}(x)$ and $\mathfrak {Im}(x)$ are the real and imaginary parts, respectively, of the complex scalar $x$.

\section{Parameter-dependent system transformations and bounded real lemmas}

In this section, we will first review the well-known Kalman-Yakubovich-Popov (KYP) Lemma and the generalized KYP lemma. Then a family of PFD mapped systems are constructed, and new PFD bounded real lemmas bounding the finite-frequency maximum singular value of a given system are presented.

\subsection{Introduction of the KYP Lemma and the Generalized KYP Lemma}
The Kalman-Yakubovich-Popov (KYP) Lemma \cite{Kal63,KYP} is a cornerstone for analyzing and synthesizing linear systems. In \cite{GKYP_Iwasaki1}, Iwasaki and Hara successfully generalized the KYP Lemma from the entire-frequency case to different finite-frequency cases. The Generalized KYP Lemma and the KYP lemma will play a fundamental role in our development.  Therefore, we state the original versions for continuous- and discrete-time LI systems in the following.
\begin{lemma}[Continuous-time KYP Lemma \cite{KYP}] \label{lem:2.1}
Consider the linear continuous-time LTI system (\ref{originalsystem}), and assume $(A,B)$ to be controllable as well as $A$ to have no eigenvalues on the imaginary axis. Given a matrix $\Pi\in\mathbb{R}^{n+m\times n+m}$, then the following statements are equivalent:
\begin{enumerate}[(1)]
\item
The frequency domain inequality
\begin{equation}\label{lemmaKYP}
\pmatset{1}{0.36pt}
  \pmatset{0}{0.2pt}
  \pmatset{2}{2pt}
  \pmatset{3}{2pt}
  \pmatset{4}{2pt}
  \pmatset{5}{2pt}
  \pmatset{6}{2pt}
\begin{pmat}[{}]
   G^*(j\omega )    \cr
   I     \cr
  \end{pmat}^*  \Pi \begin{pmat}[{}]
   G^*(j\omega )    \cr
   I     \cr
  \end{pmat} \le 0 \quad
  \text{holds for all}\quad\omega \, \in (-\infty,+\infty).
 \end{equation}
\item
There exists a symmetric matrix $P>0$ such that the following linear matrix inequality holds:
\begin{equation}\label{Dis-KYP-2}
\pmatset{1}{0.36pt}
  \pmatset{0}{0.2pt}
  \pmatset{2}{2pt}
  \pmatset{3}{2pt}
  \pmatset{4}{2pt}
  \pmatset{5}{2pt}
  \pmatset{6}{2pt}
\begin{pmat}[{.}]
   A   & B  \cr
   I   & 0 \cr
  \end{pmat}\begin{pmat}[{.}]
   0   & P  \cr
   P   & 0 \cr
  \end{pmat} \begin{pmat}[{.}]
   A   & B  \cr
   I   & 0 \cr
  \end{pmat}^* + \begin{pmat}[{.}]
   C   & D  \cr
   0   & I \cr
  \end{pmat}\Pi \begin{pmat}[{.}]
   C   & D  \cr
   0   & I \cr
  \end{pmat}^*\le 0.
\end{equation}
\item
There exist a symmetric matrix $P>0$ and matrices $K,L$ such that the following Lur'e matrix equation holds:
\begin{equation}\label{Dis-KYP-3}
\pmatset{1}{0.36pt}
  \pmatset{0}{0.2pt}
  \pmatset{2}{2pt}
  \pmatset{3}{2pt}
  \pmatset{4}{2pt}
  \pmatset{5}{2pt}
  \pmatset{6}{2pt}
\begin{pmat}[{.}]
   A   & B  \cr
   I   & 0 \cr
  \end{pmat}\begin{pmat}[{.}]
   0   & P  \cr
   P   & 0 \cr
  \end{pmat} \begin{pmat}[{.}]
   A   & B  \cr
   I   & 0 \cr
  \end{pmat}^* + \begin{pmat}[{.}]
   C   & D  \cr
   0   & I \cr
  \end{pmat}\Pi \begin{pmat}[{.}]
   C   & D  \cr
   0   & I \cr
  \end{pmat}^*=\begin{pmat}[{.}]
   -LL^*   & -LK^*  \cr
   -KL^*   & -KK^*  \cr
  \end{pmat}.
\end{equation}
\end{enumerate}
\end{lemma}

\bigskip
\begin{lemma}[Discrete-time KYP Lemma \cite{KYP}]  \label{lem:2.2}
Consider a linear discrete-time system, realized by $(A,B,C,D)$, with transfer function $G(e^{\jmath \theta})$, $(A,B)$ controllable, $A$ having no eiogenvalues of modulus 1, and a matrix $\Pi\in\mathbb{R}^{n+m\times n+m}$.  Then the following statements are equivalent:
\begin{enumerate}[(1)]
\item
The frequency domain inequality
\begin{equation}\label{lemmaGKYP}
\pmatset{1}{0.36pt}
  \pmatset{0}{0.2pt}
  \pmatset{2}{2pt}
  \pmatset{3}{2pt}
  \pmatset{4}{2pt}
  \pmatset{5}{2pt}
  \pmatset{6}{2pt}
\begin{pmat}[{}]
   G^*(e^{j\theta} )    \cr
   I     \cr
  \end{pmat}^*  \Pi \begin{pmat}[{}]
   G^*(e^{j\theta})    \cr
   I     \cr
  \end{pmat} \le 0 \quad
  \text{holds for all} \quad \theta \, \in \Theta: [-\pi,+\pi].
\end{equation}
\item
There exists a symmetric matrix $P>0$ such that the following linear matrix inequality holds:
\begin{equation}\label{GKYP}
\pmatset{1}{0.36pt}
  \pmatset{0}{0.2pt}
  \pmatset{2}{2pt}
  \pmatset{3}{2pt}
  \pmatset{4}{2pt}
  \pmatset{5}{2pt}
  \pmatset{6}{2pt}
\begin{pmat}[{.}]
   A   & B  \cr
   I   & 0 \cr
  \end{pmat}\begin{pmat}[{.}]
   0   & P  \cr
   P   & 0 \cr
  \end{pmat} \begin{pmat}[{.}]
   A   & B  \cr
   I   & 0 \cr
  \end{pmat}^* + \begin{pmat}[{.}]
   C   & D  \cr
   0   & I \cr
  \end{pmat}\Pi \begin{pmat}[{.}]
   C   & D  \cr
   0   & I \cr
  \end{pmat}^*\le 0.
\end{equation}
\item
There exist a symmetric matrix $P>0$ and matrices $K,L$ such that the following Lur'e matrix equality holds:
\begin{equation}\label{GKYP2}
\pmatset{1}{0.36pt}
  \pmatset{0}{0.2pt}
  \pmatset{2}{2pt}
  \pmatset{3}{2pt}
  \pmatset{4}{2pt}
  \pmatset{5}{2pt}
  \pmatset{6}{2pt}
\begin{pmat}[{.}]
   A   & B  \cr
   I   & 0 \cr
  \end{pmat}\begin{pmat}[{.}]
   0   & P  \cr
   P   & 0 \cr
  \end{pmat} \begin{pmat}[{.}]
   A   & B  \cr
   I   & 0 \cr
  \end{pmat}^* + \begin{pmat}[{.}]
   C   & D  \cr
   0   & I \cr
  \end{pmat}\Pi \begin{pmat}[{.}]
   C   & D  \cr
   0   & I \cr
  \end{pmat}^*=\begin{pmat}[{.}]
   -LL^*   & -LK^*  \cr
   -KL^*   & -KK^*  \cr
  \end{pmat}.
\end{equation}
\end{enumerate}
\end{lemma}

\bigskip
The generalized versions of the KYP Lemma for finite frequency ranges introduced by Iwasaki and Hara read as follows:
\begin{lemma}[Continuous-time generalized KYP lemma \cite{GKYP_Iwasaki1}]
Under the assumptions of Lemma~\ref{lem:2.1}, the following statements are equivalent:
\begin{enumerate}[(1)]
\item
The frequency domain inequality
\begin{equation}\label{Con-GKYP-1}
\pmatset{1}{0.36pt}
  \pmatset{0}{0.2pt}
  \pmatset{2}{2pt}
  \pmatset{3}{2pt}
  \pmatset{4}{2pt}
  \pmatset{5}{2pt}
  \pmatset{6}{2pt}
\begin{pmat}[{}]
   G^*(j\omega )    \cr
   I     \cr
  \end{pmat}^*  \Pi \begin{pmat}[{}]
   G^*(j\omega )    \cr
   I     \cr
  \end{pmat} \le 0 \quad
  \text{holds for all} \quad  \omega\in \Omega_l/\Omega_m/\Omega_h.
  \end{equation}
\item
There exist symmetric matrices $P$ and $Q$ of
appropriate dimensions, satisfying $Q> 0$ and
\begin{equation}\label{Con-GKYP-2}
\pmatset{1}{0.36pt}
  \pmatset{0}{0.2pt}
  \pmatset{2}{2pt}
  \pmatset{3}{2pt}
  \pmatset{4}{2pt}
  \pmatset{5}{2pt}
  \pmatset{6}{2pt}
\begin{pmat}[{.}]
   A   & B  \cr
   I   & 0 \cr
  \end{pmat} \Phi \begin{pmat}[{.}]
   A   & B  \cr
   I   & 0 \cr
  \end{pmat}^* + \begin{pmat}[{.}]
   C   & D  \cr
   0   & I \cr
  \end{pmat}\Pi \begin{pmat}[{.}]
   C   & D  \cr
   0   & I \cr
  \end{pmat}^*\le 0,
\end{equation}
where $\Phi$ is determined according to the type of frequency range considered, as shown in the following table:
\medskip
\begin{center}
\renewcommand{\arraystretch}{1.6}
\begin{tabular}{|c|c|c|}  \hline
\cline{1-3}
      \textbf{LF} (low-frequency)           & \textbf{MF} (middle-frequency)          & \textbf{HF} (high-frequency)                       \\ \hline
             $\Phi=\begin{pmat}[{.}]
   -Q   & P  \cr
    P   & \varpi_l^2 Q \cr
  \end{pmat}$    &   $\Phi=\begin{pmat}[{.}]
   -Q   & \jmath \varpi_c Q+ P  \cr
   -\jmath \varpi_c Q+ P   & \varpi_1\varpi_2Q \cr
  \end{pmat}$
     &
  $\Phi=\begin{pmat}[{.}]
   Q   &  P  \cr
   P   & -\varpi_h^2  Q \cr
  \end{pmat}$                                  \\ \hline
\end{tabular}
\end{center}
\end{enumerate}
\medskip
\end{lemma}

\bigskip
\begin{remark}\label{rem:2.1}
The main role of the KYP and GKYP lemmas is to characterize various system properties in terms of an inequality condition on the Popov function corresponding to the LTI system over the entire frequency range or over finite frequency ranges. In case the matrix $\Pi$ in (\ref{lemmaKYP}), (\ref{lemmaGKYP}) is specialized as in the common bounded-realness case:
$\Pi_{BR}= \pmatset{1}{0.36pt}
  \pmatset{0}{0.1pt}
 \pmatset{2}{1pt}
  \pmatset{3}{5pt}
  \pmatset{4}{5pt}
  \pmatset{5}{5pt}
  \pmatset{6}{5pt}\begin{pmat}[{|}]
 I &  0\cr\-
 0&  -{{\gamma ^2}I}\cr
  \end{pmat}$
or the positive-realness case:
$\Pi_{PR}=\pmatset{1}{0.36pt}
  \pmatset{0}{0.1pt}
 \pmatset{2}{1pt}
  \pmatset{3}{5pt}
  \pmatset{4}{5pt}
  \pmatset{5}{5pt}
  \pmatset{6}{5pt}\begin{pmat}[{|}]
 0 &  I\cr\-
 I&  0 \cr
  \end{pmat}$,
the (generalized) KYP lemma is referred to as {\em (generalized) bounded real lemma} or {\em (generalized) positive real lemma}.  Actually, the EF-type index (\ref{EFindex}) could be equivalently characterized by the entire-frequency inequality (\ref{lemmaKYP}) by choosing $\Pi=\Pi_{BR}$. Similarly, the FF-type index (\ref{FFindex}) could be equivalently characterized by the finite-frequency inequality (\ref{lemmaGKYP}) by choosing $\Pi=\Pi_{BR}$. For more details about the KYP and GKYP lemmas, we refer the reader to \cite{GKYP_Iwasaki1,KYP,GKYP_Pipeleers2}.
\end{remark}

\subsection{PFD Mapped Systems and PFD Bounded-Real Lemma (MF \& LF Cases)}

In this subsection, we first define a family of PFD mapped systems for a given system with respect to a pre-specified MF or LF range, then present the derived PFD bounded real lemma to show the relationships between the entire-frequency maximum singular value of the PFD mapped systems and the MF maximum singular value of the given system. Noticing that the LF range can be viewed as a special case of the MF range by letting $\varpi_c=0$ and $\varpi_d=\varpi_l$ ($\varpi_c=(\varpi_1+\varpi_2)/2, \varpi_d=(\varpi_2-\varpi_1)/2$, where the different frequencies in the LF and MF cases are defined in Table~\ref{t9}),  all the definitions and results will be presented in the more general MF setting.

\begin{definition}[PFD Mapped Systems (LF \& MF  Cases)] \label{def:2.1} Let $(A,B,C,D)$ be a realization of the LTI system \eqref{originalsystem}, $\rho\in\mathbb R$, and $\varpi_c=(\varpi_1+\varpi_2)/2, \varpi_d=(\varpi_2-\varpi_1)/2$ with $\varpi_1,\varpi_2$ defining the considered finite frequency range as in Table~\ref{t9}. Then we define the following PFD mapped systems corresponding to \eqref{originalsystem}.
\begin{enumerate}[a)]
\item The discrete-time system
$\hat G_{m \rho c} (e^{\jmath \theta}):=\pmatset{1}{0.36pt}
  \pmatset{0}{0.1pt}
 \pmatset{2}{1pt}
  \pmatset{3}{5pt}
  \pmatset{4}{5pt}
  \pmatset{5}{5pt}
  \pmatset{6}{5pt}\begin{pmat}[{|}]
  \hat{\mathbf A}_{m\rho c} &  \hat{\mathbf B}_{m\rho c} \cr\-
 \hat {\mathbf C}_{m\rho c} &  \hat{\mathbf D}_{m\rho c}\cr
  \end{pmat}$
  is constructed via the following {upper type} PFD mapping:
  $$(\hat{\mathbf A}_{m\rho c},\hat{\mathbf B}_{m\rho c},\hat{\mathbf C}_{m\rho c},\hat{\mathbf D}_{m\rho c})=\hat {\mathscr M}_{m\rho c} \left(A,B,C,D,\Omega_m\right),$$
where
\begin{equation}\label{upperPDFM}
 \begin{array}{l}
   \hat {\mathbf A}_{m\rho c}=(\rho^2+\varpi_d^2)^{\frac{1}{2}}\left( (\rho  + \jmath \varpi_c) I -A \right)^{-1}, \\
   \hat {\mathbf B}_{m\rho c}=\left( (\rho  + \jmath \varpi_c) I - A \right)^{-1} B, \\
   \hat {\mathbf C}_{m\rho c}= C \left( (\rho + \jmath \varpi_c) I - A\right)^{-1},   \\
   \hat {\mathbf D}_{m\rho c}=(\rho^2+\varpi_d^2)^{-\frac{1}{2}} \left( C\left( (\rho + \jmath \varpi_c) I -A\right)^{-1} B+D \right).
 \end{array}
\end{equation}
$\hat G_{m \rho c}$
will be referred to as \emph{upper type PFD mapped system w.r.t.\ the MF range $\Omega_m$}.
\item
The discrete-time system
$\check G_{m \rho c} (e^{\jmath \theta}):=\pmatset{1}{0.36pt}
  \pmatset{0}{0.1pt}
 \pmatset{2}{1pt}
  \pmatset{3}{5pt}
  \pmatset{4}{5pt}
  \pmatset{5}{5pt}
  \pmatset{6}{5pt}\begin{pmat}[{|}]
  \check{\mathbf A}_{m\rho c} &  \check{\mathbf B}_{m\rho c} \cr\-
  \check{\mathbf C}_{m\rho c} &  \check{\mathbf D}_{m\rho c}\cr
  \end{pmat}$
is constructed via the following lower type PFD mapping:
$$(\check{\mathbf A}_{m\rho c},\check{\mathbf B}_{m\rho c},\check{\mathbf C}_{m\rho c},\check{\mathbf D}_{m\rho c})=\check {\mathscr M}_{m\rho c} \left(A,B,C,D,\Omega_l\right),$$
where
\begin{equation}
\label{lowerPDFM}
\begin{array}{l}
   \check{\mathbf A}_{m\rho c}=   (\rho^2+\varpi_d^2)^{-\frac{1}{2}}\varpi_d(\jmath \varpi_c I -A)^{-1}\left(\varpi_d I +\frac{\rho}{\varpi_d} (\jmath \varpi_c I -A)\right), \\
   \check{\mathbf B}_{m\rho c}=   (\rho^2+\varpi_d^2)^{-\frac{1}{2}}\varpi_d(\jmath \varpi_c I -A)^{-1} B, \\
   \check{\mathbf C}_{m\rho c}=   (\rho^2+\varpi_d^2)^{-\frac{1}{2}}\varpi_d C ( \jmath \varpi_c I -A)^{-1},   \\
   \check{\mathbf D}_{m\rho c}=   (\rho^2+\varpi_d^2)^{-\frac{1}{2}} C( \jmath \varpi_c I -A)^{-1} B+(\rho^2+\varpi_d^2)^{-\frac{1}{2}}D,
 \end{array}
\end{equation}
$\check G_{l \rho c}$ will be referred to as \emph{lower type PFD mapped system w.r.t.\ the MF range $\Omega_m$}.
\item
The continuous-time system
$G_{m \rho 1} (\jmath \omega):=\pmatset{1}{0.36pt}
  \pmatset{0}{0.1pt}
 \pmatset{2}{1pt}
  \pmatset{3}{5pt}
  \pmatset{4}{5pt}
  \pmatset{5}{5pt}
  \pmatset{6}{5pt}\begin{pmat}[{|}]
  {\mathbf A}_{m\rho1} &  {\mathbf B}_{m\rho1} \cr\-
  {\mathbf C}_{m\rho1} &  {\mathbf D}_{m\rho1}\cr
  \end{pmat}$
is constructed via the following left type PFD mapping:
$$({\mathbf A}_{m\rho 1},{\mathbf B}_{m\rho 1},{\mathbf C}_{m\rho 1},{\mathbf D}_{m\rho 1})={\mathscr M}_{m\rho 1} \left(A,B,C,D,\Omega_m\right),$$
where
\begin{equation}\label{PFDML1}
\begin{array}{l}
   {\mathbf A}_{m\rho 1}=   -\frac12 I- (\rho-\jmath \varpi_d)(\jmath \varpi_1 I-A)^{-1}, \\
   {\mathbf B}_{m\rho 1}=    (\jmath \varpi_1 I -A)^{-1} B, \\
   {\mathbf C}_{m\rho 1}=    C ( \jmath \varpi_1 I -A)^{-1},   \\
   {\mathbf D}_{m\rho 1}=   -(\rho-\jmath \varpi_1 )^{-1}  \left( C (\jmath \varpi_1 I-A)^{-1} B + D \right),
\end{array}
\end{equation}
$G_{l \rho 1}$ will be referred to as \emph{left type PFD mapped system w.r.t.\ the MF range $\Omega_m$}.
\item
The continuous-time system
$G_{m \rho 2} (\jmath \omega):=\pmatset{1}{0.36pt}
  \pmatset{0}{0.1pt}
 \pmatset{2}{1pt}
  \pmatset{3}{5pt}
  \pmatset{4}{5pt}
  \pmatset{5}{5pt}
  \pmatset{6}{5pt}\begin{pmat}[{|}]
  {\mathbf A}_{m\rho2} &  {\mathbf B}_{m\rho2} \cr\-
  {\mathbf C}_{m\rho2} &  {\mathbf D}_{m\rho2}\cr
  \end{pmat}$
is constructed via the following right type PFD mapping:
$$({\mathbf A}_{m\rho 2},{\mathbf B}_{m\rho 2},{\mathbf C}_{m\rho 2},{\mathbf D}_{m\rho 1})={\mathscr M}_{m\rho 2} \left(A,B,C,D,\Omega_m\right),$$
where
\begin{equation}\label{PFDML2}
\begin{array}{l}
   {\mathbf A}_{m\rho 2}=   -\frac12 I- (\rho+\jmath \varpi_d)(\jmath \varpi_2 I-A)^{-1}, \\
   {\mathbf B}_{m\rho 2}=    (\jmath \varpi_2 I -A)^{-1} B, \\
   {\mathbf C}_{m\rho 2}=    C ( \jmath \varpi_2 I -A)^{-1},   \\
   {\mathbf D}_{m\rho 2}=   (\rho+\jmath \varpi_d )^{-1} C (\jmath \varpi_2 I-A)^{-1} B +(\rho+\jmath \varpi_1 )^{-1} D,
\end{array}
\end{equation}
$G_{l \rho 2}$ will be referred to as \emph{right type PFD mapped system w.r.t.\ the MF range $\Omega_m$}.
\end{enumerate}
\end{definition}

\bigskip
\begin{proposition} \label{prop:2.1}
Letting $\rho_m^*=\max{{\left(\varpi_d^2-\mathfrak{Re}(\lambda_i)^2-(\varpi_c+\mathfrak{Im}(\lambda_i))^2\right)} \mathord{\left/
 {\vphantom {{} {}}} \right.
 \kern-\nulldelimiterspace} {2\mathfrak{Re}(\lambda_i)}}, i=1,2,...n$, where $\lambda_i,i=1,2,...n$ are eigenvalues of the matrix $A$, then the following statements are true.\\
a). If $\rho>\rho_m^*$,  then the matrix  $\hat{\mathbf A}_{m\rho c}$ is Schur stable.    \\
b). If $\rho<-\rho_m^*$,  then the matrix $\check{\mathbf A}_{m\rho c}$ is Schur stable.  \\
c). If $\rho>\rho_m^*$,  then the matrix  ${\mathbf A}_{m\rho 1}$ is Hurwitz stable.      \\
d). If $\rho>\rho_m^*$,  then the matrix  ${\mathbf A}_{m\rho 2}$ is Hurwitz stable.      \\
\end{proposition}
 \noindent \textbf{Proof.} a). From the upper case PFD mapping (\ref{upperPDFM}), the eigenvalues $\hat \lambda_{m\rho ci}, i=1,...,n$ of the mapped matrix $\hat{\mathbf A}_{m\rho c}$ are:
$\hat \lambda_{m\rho ci}=(\rho^2+\varpi_d^2)^{\frac{1}{2}}(\rho  + \jmath \varpi_c  -\lambda_i)^{-1}$.  If $\rho>\rho_m^*$,  we have  $\left| \hat \lambda_{m\rho ci} \right|<1, i=1,...,n$. Thus the matrix $\hat{\mathbf A}_{m\rho c}$ is Schur stable.\\
Similarly, the statements b)-d) could be proved by observing the eigenvalues of the mapped matrices. \\

\begin{theorem} \label{thm:2.1} (\emph{PFD Bounded-Real Lemma (LF\&\& MF Case)}) Denote the entire-frequency range ($\theta \in [-\pi, +\pi]$) in the discrete-time setting as $\Theta$, and use $\Omega$ and $\Omega_m$ to represent the entire-frequency range and middle-frequency range (see Table \uppercase\expandafter{\romannumeral1}), respectively. The following statements on the relationship between the maximum singular value of the mapped systems over entire-frequency range and the maximum singular value of the given system are true:\\
a). If $\sigma_{\max}(\hat G_{m \rho c}({e^{\jmath \theta}})) \leq \hat \gamma_{m \rho c}, \forall {\theta \in \Theta}$,  then $\sigma_{\max}\left(G(\jmath \omega)\right) \leq (\rho^2+\varpi_d^2)^{\frac{1}{2}} \hat \gamma_{m \rho c}, \forall \omega \in \Omega_m$.

\noindent b). If $\sigma_{\max}(\check G_{m \rho c}({e^{\jmath \theta}})) \leq \check \gamma_{m \rho c}, \forall {\theta \in \Theta}$,  then $\sigma_{\max}\left(G(\jmath \omega)\right) \leq  (\rho^2+\varpi_d^2)^{\frac{1}{2}} \check \gamma_{m \rho c}, \forall \omega \in \Omega_m$.

\noindent c). If $\sigma_{\max}\left(G_{m \rho 1}({\jmath \omega} )\right) \leq \gamma_{m \rho 1}, \forall {\omega \in \Omega}$, then $\sigma_{\max}\left(G(\jmath \omega)\right) \leq  (\rho^2+\varpi_d^2)^{\frac{1}{2}} \gamma_{m \rho 1}, \forall \omega \in \Omega_m$.

\noindent d). If $\sigma_{\max}\left(G_{m \rho 2}({\jmath \omega})\right) \leq \gamma_{m \rho 2}, \forall {\omega \in \Omega}$, then $\sigma_{\max}\left(G(\jmath \omega)\right) \leq  (\rho^2+\varpi_d^2)^{\frac{1}{2}}  \gamma_{m \rho 2}, \forall \omega \in \Omega_m$.\\
\end{theorem}

\noindent \textbf{Proof.} a). Since $\sigma_{\max}(\hat G_{m \rho c}({e^{\jmath \theta}})) \leq \hat \gamma_{m \rho c}, \forall {\theta \in \Theta: [-\pi, +\pi]}$ equalivent to
\begin{equation}
\label{lemma11}
\pmatset{1}{0.36pt}
  \pmatset{0}{0.2pt}
  \pmatset{2}{2pt}
  \pmatset{3}{2pt}
  \pmatset{4}{2pt}
  \pmatset{5}{2pt}
  \pmatset{6}{2pt}
\begin{pmat}[{}]
   \hat G^*_{m \rho c}({e^{\jmath \theta}})    \cr
   I     \cr
  \end{pmat}^*  \begin{pmat}[{.}]
   I & 0  \cr
   0 & -\hat {\gamma}^2_{m\rho c}I     \cr
  \end{pmat} \begin{pmat}[{}]
   \hat G^*_{m \rho c}({e^{\jmath \theta}})   \cr
   I     \cr
  \end{pmat} \le 0, \forall \theta {\kern 4pt} \in \Theta: [-\pi,+\pi].
  \end{equation}

\noindent According to the discrete-time KYP lemma, there exists a positive symmetrical matrix $\hat {\mathbf P}_{m\rho c}$ and $\hat {\mathbf L}_{m\rho c}, \hat {\mathbf K}_{m\rho c}$ satisfying
\begin{spacing}{1.0}
\begin{subequations}\label{KYPlurD1}
   \begin{align}
&  \hat{\mathbf A}_{m\rho c} \hat {\mathbf P}_{m\rho c} \hat{\mathbf A}^*_{m\rho c}  - \hat {\mathbf P}_{m\rho c}         +  \hat{\mathbf B}_{m\rho c}\hat{\mathbf B}^*_{m\rho c}=-\hat {\mathbf L}_{m\rho c}\hat {\mathbf L}_{m\rho c}^*,   \\
&  \hat{\mathbf A}_{m\rho c} \hat{\mathbf P}_{m\rho c}\hat{\mathbf C}^*_{m\rho c} +  \hat{\mathbf B}_{m\rho c}\hat{\mathbf D}^*_{m\rho c} =-\hat {\mathbf L}_{m\rho c}\hat {\mathbf K}_{m\rho c}^* \\
&  \hat{\mathbf C}_{m\rho c} \hat{\mathbf P}_{m\rho c}\hat{\mathbf C}^*_{m\rho c} +   \hat{\mathbf D}_{m\rho c}\hat{\mathbf D}^*_{m\rho c}-\hat \gamma^2_{m\rho c}I=-\hat {\mathbf K}_{m\rho c}\hat {\mathbf K}_{m\rho c}^*,
   \end{align}
\end{subequations}
\end{spacing}

\noindent Define $Q=\hat {\mathbf P}_{m\rho c}, P=\rho \hat {\mathbf P}_{m\rho c}$, from the above equation (\ref{KYPlurD1}a-\ref{KYPlurD1}c) we have
\begin{equation}
\begin{array}{l}
             {\kern 10pt}-\mathbf{He}((\jmath \varpi_1 I-  A)  Q  (\jmath \varpi_2 I- A))    +  AP +  PA^* + BB^*   \\
            =  -(\jmath \varpi_c I-  A)  \hat {\mathbf P}_{m\rho c}  (\jmath \varpi_c I-    A )^*        + \varpi_d^2  \hat {\mathbf P}_{m\rho c}  -\rho ( \jmath \varpi_c I-A) \hat {\mathbf P}_{m\rho c}  -\rho \hat {\mathbf P}_{m\rho c} ( \jmath \varpi_c I-A)^*   + B B^*   \\
              =  (\rho^2+\varpi_d^2)  \hat {\mathbf P}_{m\rho c} -(\rho I +\jmath \varpi_c I-  A)  \hat {\mathbf P}_{m\rho c}  (\rho I+ \jmath \varpi_c I-    A )^*         + B B^*   \\
            = (\rho I+\jmath \varpi_c I-  A)\left\{  \hat{\mathbf A}_{m\rho c} \hat {\mathbf P}_{m\rho c} \hat{\mathbf A}^*_{m\rho c}  - \hat {\mathbf P}_{m\rho c}         +  \hat{\mathbf B}_{m\rho c}\hat{\mathbf B}^*_{m\rho c}  \right\}(\rho I+\jmath \varpi_c I-  A)^*    \\
            \mathop {=}\limits^{\ref{KYPlurD1}a}  (\rho I+\jmath \varpi_c I-  A)\left\{ -\hat {\mathbf L}_{m\rho c}\hat {\mathbf L}_{m\rho c}^* \right\}(\rho I+\jmath \varpi_c I-  A)^*    \\
                    \end{array}
\end{equation}

\begin{equation}
\begin{array}{l}
             {\kern 10pt} (\jmath \varpi_c I-  A)  Q C^*    +     PC^* + BD^*   \\
            =  (\jmath \varpi_c I-  A)  \hat {\mathbf P}_{m\rho c} C^*    +  \rho \hat {\mathbf P}_{m\rho c} C^*  +   B D^*   \\
            = (\rho I+\jmath \varpi_c I-  A)\hat {\mathbf P}_{m\rho c} C^*   +   B D^*      \\
             \mathop {=}\limits^{15a} (\rho I+\jmath \varpi_c I-  A)\left\{ \begin{array}{l}
+ (\rho^2+\varpi_d^2) (\rho I+\jmath \varpi_c I-  A)^{-1}\hat {\mathbf P}_{m\rho c} (\rho I+\jmath \varpi_c I-  A)^{-*} \\
+ (\rho I+\jmath \varpi_c I-  A)^{-1}BB^* (\rho I+\jmath \varpi_c I-  A)^{-*} \\
+ \hat {\mathbf L}_{m\rho c}\hat {\mathbf L}_{m\rho c}^* \\
 \end{array} \right\} C^*+BD^* \\
= {(\rho^2+\varpi_d^2)^{\frac{1}{2}}} (\rho I+\jmath \varpi_c I-  A)  \left\{ \hat{\mathbf A}_{m\rho c} \hat {\mathbf P}_{m\rho c}\hat{\mathbf C}^*_{m\rho c} +  \hat{\mathbf B}_{m\rho c}\hat{\mathbf D}^*_{m\rho c} \right\}+(\rho I+\jmath \varpi_c I-  A) \hat {\mathbf L}_{m\rho c}\hat {\mathbf L}_{m\rho c}^* C^*   \\
 \mathop {=}\limits^{\ref{KYPlurD1}b} (\rho I+\jmath \varpi_c I-  A) \left\{ -{(\rho^2+\varpi_d^2)^{\frac{1}{2}}} \hat {\mathbf L}_{m\rho c}\hat {\mathbf K}_{m\rho c}^*+\hat {\mathbf L}_{m\rho c}\hat {\mathbf L}_{m\rho c}^*C^*\right\} \\
  =- (\rho I+\jmath \varpi_c I-  A)\hat {\mathbf L}_{m\rho c}\left({(\rho^2+\varpi_d^2)^{\frac{1}{2}}}  \hat {\mathbf K}_{m\rho c} -C \hat {\mathbf L}_{m\rho c} \right)^*\\
         \end{array}
\end{equation}

\begin{equation}
\begin{array}{l}
             {\kern 10pt} -C Q C^*    +   DD^*-(\rho^2+\varpi_d^2)\hat {\gamma}_{m\rho c}^2 I  \\
            \mathop {=}\limits^{\ref{KYPlurD1}a} -C\left\{ \begin{array}{l}
+ (\rho^2+\varpi_d^2) (\rho I+\jmath \varpi_c I-  A)^{-1}\hat {\mathbf P}_{m\rho c} (\rho I+\jmath \varpi_c I-  A)^{-*} \\
+ (\rho I+\jmath \varpi_c I-  A)^{-1}BB^* (\rho I+\jmath \varpi_c I-  A)^{-*} \\
+ \hat {\mathbf L}_{m\rho c}\hat {\mathbf L}_{m\rho c}^* \\
 \end{array} \right\} C^*+  DD^*-(\rho^2+\varpi_d^2)\hat {\gamma}_{m\rho c}^2 I  \\
      \mathop {=}\limits^{\ref{KYPlurD1}a} (\rho^2+\varpi_d^2) \left\{ \begin{array}{l}
+ C(\rho I+\jmath \varpi_c I-  A)^{-1}\hat {\mathbf P}_{m\rho c} (\rho I+\jmath \varpi_c I-  A)^{-*}C^* \\
+ (\rho^2+\varpi_d^2)^{-1} C(\rho I+\jmath \varpi_c I-  A)^{-1}BB^* (\rho I+\jmath \varpi_c I-  A)^{-*}   C^*\\
+ (\rho^2+\varpi_d^2)^{-1} C (\rho I+\jmath \varpi_c I-  A)^{-1} BD^*   \\
+ (\rho^2+\varpi_d^2)^{-1} D B^* (\rho I+\jmath \varpi_c I-  A)^{-*}  C^*\\
+ (\rho^2+\varpi_d^2)^{-1} DD^* \\
- \hat {\gamma}_{m\rho c}^2 I \\
 \end{array} \right\} \\
-(\rho^2+\varpi_d^2) \left\{ \begin{array}{l}
+ 2C (\rho I+\jmath \varpi_c I-  A)^{-1}\hat {\mathbf P}_{m\rho c} (\rho I+\jmath \varpi_c I-  A)^{-*} C^*\\
+ 2(\rho^2+\varpi_d^2)^{-1} C(\rho I+\jmath \varpi_c I-  A)^{-1}BB^* (\rho I+\jmath \varpi_c I-  A)^{-*} C^*\\
+ (\rho^2+\varpi_d^2)^{-1} C (\rho I+\jmath \varpi_c I-  A)^{-1} BD^* \\
+ (\rho^2+\varpi_d^2)^{-1} D B^* (\rho I+\jmath \varpi_c I-  A)^{-*}  C^*     \\
 \end{array} \right\} \\
 -C\hat {\mathbf L}_{m\rho c}\hat {\mathbf L}_{m\rho c}^*C^*\\
    \mathop {===}\limits^{\ref{KYPlurD1}b,\ref{KYPlurD1}c} (\rho^2+\varpi_d^2)  \left\{\hat{\mathbf C}_{l \rho c} \hat {\mathbf P}_{m\rho c}\hat{\mathbf C}^*_{m\rho c} +   \hat{\mathbf D}_{m\rho c}\hat{\mathbf D}^*_{m\rho c}-\hat \gamma^2_{m\rho c}I\right\} \\
  {\kern 12pt} +(\rho^2+\varpi_d^2) ^{\frac{1}{2}}  C \left\{\varpi_d^{-1}(\hat {\mathbf A}_{m\rho c} \hat{\mathbf P}_{m\rho c}\hat{\mathbf C}^*_{m\rho c} + \hat{\mathbf B}_{m\rho c}\hat{\mathbf D}^*_{m\rho c} )\right\}\\
  {\kern 12pt}+(\rho^2+\varpi_d^2) ^{\frac{1}{2}} \left\{\varpi_d^{-1}(\hat {\mathbf A}_{m\rho c} \hat{\mathbf P}_{m\rho c}\hat{\mathbf C}^*_{m\rho c} + \hat{\mathbf B}_{m\rho c}\hat{\mathbf D}^*_{m\rho c} )\right\} C^* \\
  {\kern 12pt} -  C\hat {\mathbf L}_{m\rho c}\hat {\mathbf L}_{m\rho c}^*C^*\\
    =-(\rho^2+\varpi_d^2)    \hat {\mathbf K}_{m\rho c} \hat {\mathbf K}_{m\rho c}^* +(\rho^2+\varpi_d^2)^{\frac{1}{2}}C\hat {\mathbf L}_{m\rho c}\hat {\mathbf K}_{m\rho c}^*+ (\rho^2+\varpi_d^2)^{\frac{1}{2}}\hat {\mathbf K}_{m\rho c}\hat {\mathbf L}_{m\rho c}C^* -C\hat {\mathbf L}_{m\rho c}\hat {\mathbf L}_{m\rho c}^*C^* \\
      =-\left((\rho^2+\varpi_d^2)^{\frac{1}{2}}\hat {\mathbf K}_{m\rho c}-C\hat {\mathbf L}_{m\rho c}\right)\left((\rho^2+\varpi_d^2)^{\frac{1}{2}}\hat {\mathbf K}_{m\rho c}-C\hat {\mathbf L}_{m\rho c}\right)^*
         \end{array}
\end{equation}

\noindent Combing the above equations, we have:
 \begin{equation}
\begin{array}{l}
 \pmatset{1}{0.1pt}
  \pmatset{0}{0.1pt}
  \pmatset{2}{4pt}
  \pmatset{3}{4pt}
  \pmatset{4}{4pt}
  \pmatset{5}{4pt}
  \pmatset{6}{1pt}
  {\kern 10pt}
\begin{pmat}[{.}]
   A   & I \cr
   C   & 0 \cr
  \end{pmat}\begin{pmat}[{.}]
   -Q   & +\jmath \varpi_c Q+ P  \cr
  -\jmath \varpi_c Q+P   & \varpi_1\varpi_2Q \cr
  \end{pmat} \begin{pmat}[{.}]
   A   & I  \cr
   C   & 0 \cr
  \end{pmat}^* + \begin{pmat}[{.}]
   B   & 0  \cr
   D   & I \cr
  \end{pmat}\begin{pmat}[{.}]
   I   & 0  \cr
   0   & -(\rho^2+\varpi_d^2)  \hat {\gamma}_{m\rho c}^2 \cr
  \end{pmat} \begin{pmat}[{.}]
   B   & 0  \cr
   D   & I \cr
  \end{pmat}^*  \\
 = {\begin{pmat}[{|}]
  -\mathbf{He}((\jmath \varpi_1 I-  A)  Q  (\jmath \varpi_2 I- A))    +  AP +  PA^* + BB^*   & (\jmath \varpi_c I-  A)  Q C^*    +     PC^* + BD^*  \cr\-
      *  &  -C Q C^*    +   DD^*-(\rho^2+\varpi_d^2)  \hat {\gamma}_{m\rho c}^2 I \cr
  \end{pmat}} \\[4mm]
=  {\begin{pmat}[{|}]
     -LL^*      & - LK^*  \cr\-
      *             & -KK^*   \cr
  \end{pmat}} \\
 \end{array}\end{equation}

\noindent  where
\[\begin{array}{l}
L=(\rho I+\jmath \varpi_c I-  A) \hat {\mathbf L}_{m\rho c} \\
K= C\hat {\mathbf L}_{m\rho c}-(\rho^2+\varpi_d^2)^{\frac{1}{2}}\hat {\mathbf K}_{m\rho c}\\
 \end{array}\]

\noindent According to the GKYP lemma (Lemma 2.3), the following inequality can be concluded:
\begin{equation}
\label{lemma11}
\pmatset{1}{0.36pt}
  \pmatset{0}{0.2pt}
  \pmatset{2}{2pt}
  \pmatset{3}{2pt}
  \pmatset{4}{2pt}
  \pmatset{5}{2pt}
  \pmatset{6}{2pt}
\begin{pmat}[{}]
   G^*(j\omega )    \cr
   I     \cr
  \end{pmat}^*  \begin{pmat}[{|}]
    {\kern 4pt} I  {\kern 4pt}   &  0 \cr\-
   0     & -(\rho^2+\varpi_d^2)  \hat {\gamma}_{m\rho c}^2  \cr
  \end{pmat}  \begin{pmat}[{}]
   G^*(j\omega )    \cr
   I     \cr
  \end{pmat} \le 0,  \forall \omega  \in \Omega_l: [\varpi_1,\varpi_2]. \end{equation}
\noindent This leads to

\begin{equation}
\sigma_{\max}\left(G(\jmath\omega)\right) \leq (\rho^2+\varpi_d^2)^{\frac{1}{2}} \hat {\gamma}_{m\rho c},   \forall \omega  \in  \Omega_m: [\varpi_1,\varpi_2].
\end{equation}
\noindent this completes the proof of statement (a).\\

\noindent b). Since $\sigma_{\max}\left(\check G_{l \rho c}({e^{\jmath \theta}})\right) \leq \check \gamma_{l \rho c}, \forall {\theta \in \Theta}$ is equivalent to
\begin{equation}
\label{lemma11}
\pmatset{1}{0.36pt}
  \pmatset{0}{0.2pt}
  \pmatset{2}{2pt}
  \pmatset{3}{2pt}
  \pmatset{4}{2pt}
  \pmatset{5}{2pt}
  \pmatset{6}{2pt}
\begin{pmat}[{}]
   G^*(e^{\jmath\theta} )    \cr
   I     \cr
  \end{pmat}^*  \Pi \begin{pmat}[{}]
   G^*(e^{\jmath\theta} )    \cr
   I     \cr
  \end{pmat} \le 0,{\kern 4pt}
  \forall \theta {\kern 4pt} \in [-\pi,+\pi].
  \end{equation}
\noindent According to the discrete-time KYP lemma (Lemma 2.2), there exists a positive symmetrical matrix $\check {\mathbf P}_{m\rho c}$ and $\check{\mathbf L}_{m\rho c}, \check {\mathbf K}_{m\rho c}$ satisfying
\begin{spacing}{1.2}
\begin{subequations}\label{KYPlurD2}
   \begin{align}
&  \check{\mathbf A}_{m\rho c} \check {\mathbf P}_{m\rho c} \check{\mathbf A}^*_{m\rho c}  - \check {\mathbf P}_{m\rho c}         +  \check{\mathbf B}_{m\rho c}\check{\mathbf B}^*_{m\rho c}=-\check {\mathbf L}_{m\rho c}\check {\mathbf L}_{m\rho c}^*,   \\
&  \check{\mathbf A}_{m\rho c} \check {\mathbf P}_{m\rho c}\check{\mathbf C}^*_{m\rho c} +  \check{\mathbf B}_{m\rho c}\check{\mathbf D}^*_{m\rho c} =-\check {\mathbf L}_{m\rho c}\check {\mathbf K}_{m\rho c}^* \\
&  \check{\mathbf C}_{m\rho c} \check {\mathbf P}_{m\rho c}\check{\mathbf C}^*_{m\rho c} +   \check{\mathbf D}_{m\rho c}\hat{\mathbf D}^*_{m\rho c}-\check \gamma^2_{m\rho c}I=-\check {\mathbf K}_{m\rho c}\check {\mathbf K}_{m\rho c}^*,
   \end{align}
\end{subequations}
\end{spacing}

\noindent Define $Q=\check {\mathbf P}_{m\rho c}, P=\rho \check {\mathbf P}_{m\rho c}$, from the above equation (\ref{KYPlur'e-D}) we have
\begin{equation}
\begin{array}{l}
             {\kern 10pt}-\mathbf{He}((\jmath \varpi_1 I-  A)  Q  (\jmath \varpi_2 I- A))    +  AP +  PA^* + BB^*   \\
            =  -(\jmath \varpi_c I-  A)  \check{\mathbf P}_{m\rho c}  (\jmath \varpi_c I-    A )^*        + \varpi_d^2  \check{\mathbf P}_{m\rho c}  -\rho ( \jmath \varpi_c I-A) \check {\mathbf P}_{m\rho c}  -\rho  \check {\mathbf P}_{m\rho c} ( \jmath \varpi_c I-A)^*   + B B^*   \\
            =(\varpi_d I +\frac{\rho}{\varpi_d} (\jmath \varpi_c I -A))\check {\mathbf P}_{m\rho c}(\varpi_d I +\frac{\rho}{\varpi_d} (\jmath \varpi_c I -A))^* \\ {\kern 10pt} -(\rho^2+\varpi_d^2)\varpi_d^{-2}(\jmath \varpi_c I-  A) \check {\mathbf P}_{m\rho c} (\jmath \varpi_c I-  A)^*+BB^*\\
            = (\rho^2+\varpi_d^2)\varpi_d^{-2} (\jmath \varpi_c I-  A)\left\{  \check{\mathbf A}_{m\rho c} \check {\mathbf P}_{m\rho c} \check{\mathbf A}^*_{m\rho c}  - \check{\mathbf P}_{m\rho c}         +  \check{\mathbf B}_{m\rho c}\check{\mathbf B}^*_{m\rho c}  \right\}(\jmath \varpi_c I-  A)^*    \\
            \mathop {=}\limits^{\ref{KYPlurD2}a} (\rho^2+\varpi_d^2)\varpi_d^{-2}(\jmath \varpi_c I-  A) \left\{ -\check{\mathbf L}_{m\rho c}\check{\mathbf L}_{m\rho c}^* \right\}(\jmath \varpi_c I-  A)^*    \\
                    \end{array}
\end{equation}
\begin{equation}
\begin{array}{l}
             {\kern 10pt} (\jmath \varpi_c I-  A)  Q C^*    +     PC^* + BD^*   \\
             =(\jmath \varpi_c I-  A)  \check {\mathbf P}_{m\rho c} C^*    +     \rho \varpi_d \check {\mathbf P}_{m\rho c}C^* + BD^* \\
             \mathop{=}\limits^{\ref{KYPlurD2}a} (\jmath \varpi_c I-  A)\left\{ \begin{array}{l}
+ \varpi_d^2 (\jmath \varpi_c I-  A)^{-1}\check {\mathbf P}_{m\rho c} (\jmath \varpi_c I-  A)^{-*} \\
+ (\jmath \varpi_c I-  A)^{-1}BB^* (\jmath \varpi_c I-  A)^{-*} \\
- \rho  (\jmath \varpi_c I-  A)^{-1} \check {\mathbf P}_{m\rho c} \\
- \rho   \check {\mathbf P}_{m\rho c} (\jmath \varpi_c I-  A)^{-*} \\
+ (\rho^2+\varpi_d^2)\varpi_d^{-2}\check {\mathbf L}_{m\rho c}\check {\mathbf L}_{m\rho c}^* \\
 \end{array} \right\} C^*+     \rho \varpi_d \check {\mathbf P}_{m\rho c}C^* + BD^* \\
=  \varpi_d(\varpi_d-\frac{\rho}{\varpi_d} (\jmath \varpi_c I-  A)) \check {\mathbf P}_{m\rho c} (\jmath \varpi_c I-  A)^{-*}C^* + B(B^*(\jmath \varpi_c I-  A)^{-*}C^*+D^*)\\
 \mathop {=}\limits^{\ref{KYPlurD2}b} (\rho ^2+\varpi_d^2)\varpi_d^{-2} \varpi_d (\jmath \varpi_c I-  A) \left\{  \check{\mathbf A}_{lc\rho} \check {\mathbf P}_{m\rho c}\check{\mathbf C}^*_{m\rho c} +  \check{\mathbf B}_{m\rho c}\check{\mathbf D}^*_{m\rho c} \right\}\\
  {\kern 12pt}+    (\rho ^2+\varpi_d^2)\varpi_d^{-2}(\jmath \varpi_c I-  A)\check {\mathbf L}_{m\rho c}\check {\mathbf L}_{m\rho c}^*C^*  \\
  =-  (\rho ^2+\varpi_d^2)\varpi_d^{-2}(\jmath \varpi_c I-  A) \check {\mathbf L}_{m\rho c}\left( \varpi_d \check {\mathbf K}_{m\rho c}-C \check {\mathbf L}_{m\rho c}   \right)^*\\
         \end{array}
\end{equation}
\begin{equation}
\begin{array}{l}
             {\kern 10pt} -C Q C^*    +   DD^*-(\rho^2+\varpi_d^{2})  \hat {\gamma}_{m\rho c}^2I   \\
             = -C \check {\mathbf P}_{m\rho c} C^*    +   DD^*-(\rho^2+\varpi_d^{2})  \hat {\gamma}_{m\rho c}^2 I   \\
            \mathop {==}\limits^{22a} -C\left\{ \begin{array}{l}
+ \varpi_d^2 (\jmath \varpi_c I-  A)^{-1}\check {\mathbf P}_{m\rho c} (\jmath \varpi_c I-  A)^{-*} \\
+ (\jmath \varpi_c I-  A)^{-1}BB^* (\jmath \varpi_c I-  A)^{-*} \\
- \rho  (\jmath \varpi_c I-  A)^{-1} \check {\mathbf P}_{m\rho c} \\
- \rho   \check {\mathbf P}_{m\rho c} (\jmath \varpi_c I-  A)^{-*} \\
+ (\rho ^2+\varpi_d^2)\varpi_d^{-2} \check {\mathbf L}_{m\rho c}\check {\mathbf L}_{m\rho c}^* \\
 \end{array} \right\} C^* +  DD^*-(\rho ^2+\varpi_d^2) \hat {\gamma}_{m\rho c}^2 I \\
      = (\rho ^2+\varpi_d^2)\varpi_d^{-2} \varpi_d^2\left\{ \begin{array}{l}
+ (\rho ^2+\varpi_d^2)^{-1} \varpi_d^{2} C(\jmath \varpi_c I-  A)^{-1}\check {\mathbf P}_{m\rho c} (\jmath \varpi_c I-  A)^{-*}C^* \\
+ (\rho ^2+\varpi_d^2)^{-1} C(\jmath \varpi_c I-  A)^{-1}BB^* (\jmath \varpi_c I-  A)^{-*}   C^*\\
+ (\rho ^2+\varpi_d^2)^{-1} C (\jmath \varpi_c I-  A)^{-1} BD^*   \\
+ (\rho ^2+\varpi_d^2)^{-1} D B^* (\jmath \varpi_c I-  A)^{-*}  C^*\\
+ (\rho ^2+\varpi_d^2)^{-1}  DD^* \\
- \hat {\gamma}_{m\rho c}^2 I \\
 \end{array} \right\} \\
-(\rho ^2+\varpi_d^2)  \left\{ \begin{array}{l}
+ 2(\rho ^2+\varpi_d^2)^{-1} \varpi_d^{2}C (\jmath \varpi_c I-  A)^{-1}\hat {\mathbf P}_{m\rho c} (\jmath \varpi_c I-  A)^{-*} C^*\\
+ 2(\rho ^2+\varpi_d^2)^{-1}  C(\jmath \varpi_c I-  A)^{-1}BB^* (\jmath \varpi_c I-  A)^{-*} C^*\\
- (\rho ^2+\varpi_d^2)^{-1}  \rho (\jmath \varpi_c I-  A)^{-1} \check {\mathbf P}_{m\rho c} \\
- (\rho ^2+\varpi_d^2)^{-1} \rho  \check {\mathbf P}_{m\rho c} (\jmath \varpi_c I-  A)^{-*} \\
+ (\rho ^2+\varpi_d^2)^{-1}  C (\jmath \varpi_c I-  A)^{-1} BD^*   \\
+ (\rho ^2+\varpi_d^2)^{-1} D B^* (\jmath \varpi_c I-  A)^{-*}  C^*\\
 \end{array} \right\} \\
 -(\rho ^2+\varpi_d^2)\varpi_d^{-2} C\check {\mathbf L}_{m\rho c}\check {\mathbf L}_{m\rho c}^*C^*\\
    \mathop {===}\limits^{\ref{KYPlurD2}b,\ref{KYPlurD2}c}(\rho ^2+\varpi_d^2)   \left\{\check{\mathbf C}_{lc\rho} \check{\mathbf P}_{m\rho c}\check{\mathbf C}^*_{m\rho c} +   \check{\mathbf D}_{m\rho c}\check{\mathbf D}^*_{m\rho c}-\check{ \gamma}^2_{m\rho c}I\right\}\\
  {\kern 12pt} +(\rho ^2+\varpi_d^2)  C \left\{\varpi_d^{-1}(\check{\mathbf A}_{lc\rho} \check{\mathbf P}_{m\rho c}\check{\mathbf C}^*_{m\rho c} + \check{\mathbf B}_{m\rho c}\check{\mathbf D}^*_{m\rho c} )\right\}\\
  {\kern 12pt} + (\rho ^2+\varpi_d^2)   \left\{\varpi_d^{-1}  (\check{\mathbf A}_{lc\rho} \check{\mathbf P}_{m\rho c}\check{\mathbf C}^*_{m\rho c} + \check{\mathbf B}_{m\rho c}\check{\mathbf D}^*_{m\rho c} ) \right\}C^* \\
  {\kern 12pt} -(\rho ^2+\varpi_d^2)\varpi_d^{-2} C\check {\mathbf L}_{m\rho c}\check {\mathbf L}_{m\rho c}^*C^*\\
   =-(\rho ^2+\varpi_d^2) \check {\mathbf K}_{m\rho c}\check {\mathbf K}_{m\rho c}^* +(\rho ^2+\varpi_d^2)\varpi_d^{-2} \varpi_d   C\check {\mathbf L}_{m\rho c}\check {\mathbf K}_{m\rho c}^*\\
     {\kern 12pt} + (\rho ^2+\varpi_d^2)\varpi_d^{-2} \varpi_d \check {\mathbf K}_{m\rho c}\hat {\mathbf L}_{m\rho c}C^* -C\check {\mathbf L}_{m\rho c}\check {\mathbf L}_{m\rho c}^*C^*\\
  = -(\rho ^2+\varpi_d^2)\varpi_d^{-2} \left(\varpi_d  \check {\mathbf K}_{m\rho c}-C\check {\mathbf L}_{m\rho c}\right) \left(\varpi_d  \check {\mathbf K}_{m\rho c}-C\check {\mathbf L}_{m\rho c}\right)^*
         \end{array}
\end{equation}

\noindent Combing the above equations, we have:
 \begin{equation}
\begin{array}{l}
 \pmatset{1}{0.1pt}
  \pmatset{0}{0.1pt}
  \pmatset{2}{4pt}
  \pmatset{3}{4pt}
  \pmatset{4}{4pt}
  \pmatset{5}{4pt}
  \pmatset{6}{1pt}
  {\kern 10pt}
\begin{pmat}[{.}]
   A   & I \cr
   C   & 0 \cr
  \end{pmat}\begin{pmat}[{.}]
   -Q   & +\jmath \varpi_c Q+ P  \cr
  -\jmath \varpi_c Q+P   & \varpi_1\varpi_2Q \cr
  \end{pmat} \begin{pmat}[{.}]
   A   & I \cr
   C   & 0 \cr
  \end{pmat}^* + \begin{pmat}[{.}]
   B   & 0 \cr
   D   & I \cr
  \end{pmat}\begin{pmat}[{|}]
    {\kern 4pt} I  {\kern 4pt}   & 0  \cr\-
   0   & -(\rho ^2+\varpi_d^2)   \hat {\gamma}_{m\rho c}^2I \cr
  \end{pmat} \begin{pmat}[{.}]
   B   & 0  \cr
   D   & I \cr
  \end{pmat}^*  \\
 = {\begin{pmat}[{|}]
  -\mathbf{He}((\jmath \varpi_1 I-  A)  Q  (\jmath \varpi_2 I- A))    +  AP +  PA^* + BB^*   & (\jmath \varpi_c I-  A)  Q C^*    +     PC^* + BD^*  \cr\-
      *  &  -C Q C^*    +   DD^*-(\rho ^2+\varpi_d^2)  \hat {\gamma}_{m\rho c}^2 I \cr
  \end{pmat}} \\[4mm]
=  {\begin{pmat}[{|}]
     -LL^*      & - LK^*  \cr\-
      *             & -KK^*   \cr
  \end{pmat}} \\
 \end{array}\end{equation}

\noindent  where
\[\begin{array}{l}
L=(\rho ^2+\varpi_d^2)^{\frac{1}{2}}\varpi_d^{-1} (\jmath \varpi_c I-  A) \check {\mathbf L}_{m\rho c} \\
K=(\rho ^2+\varpi_d^2)^{\frac{1}{2}}\varpi_d^{-1}  \left(  \varpi_d  \check {\mathbf K}_{m\rho c}- C\check {\mathbf L}_{m\rho c}\right)\\
 \end{array}\]

\noindent According to the GKYP lemma (Lemma 2.3), the following inequality can be concluded:

\begin{equation}
\label{lemma11}
\pmatset{1}{0.36pt}
  \pmatset{0}{0.2pt}
  \pmatset{2}{2pt}
  \pmatset{3}{2pt}
  \pmatset{4}{2pt}
  \pmatset{5}{2pt}
  \pmatset{6}{2pt}
\begin{pmat}[{}]
   G^*(j\omega )    \cr
   I     \cr
  \end{pmat}^*  \begin{pmat}[{|}]
    {\kern 4pt} I  {\kern 4pt}   &  0 \cr\-
   0     &  -(\rho ^2+\varpi_d^2) \hat {\gamma}_{m\rho c}^2 \cr
  \end{pmat}  \begin{pmat}[{}]
   G^*(j\omega )    \cr
   I     \cr
  \end{pmat} \le 0,
 holds {\kern 4pt} for {\kern 4pt} all{\kern 4pt} \omega {\kern 4pt} \in [\varpi_1,\varpi_2]. \end{equation}

\noindent This leads to\begin{equation}
\sigma_{\max}\left(G(\jmath\omega)\right) \leq (\rho ^2+\varpi_d^2) ^{\frac{1}{2}} \check {\gamma}_{m\rho c},   \forall \omega  \in  \Omega_m: [\varpi_1,\varpi_2].
\end{equation}
\noindent c). Since $\sigma_{\max}\left(G_{l \rho 1}({ \jmath \omega})\right) \leq   \gamma_{l \rho 1}, \forall {\omega \in \Omega: [-\infty, +\infty]}$ equivalent to

\begin{equation}
\label{lemma11}
\pmatset{1}{0.36pt}
  \pmatset{0}{0.2pt}
  \pmatset{2}{2pt}
  \pmatset{3}{2pt}
  \pmatset{4}{2pt}
  \pmatset{5}{2pt}
  \pmatset{6}{2pt}
\begin{pmat}[{}]
   G^*(\jmath \omega )    \cr
   I     \cr
  \end{pmat}^*  \Pi \begin{pmat}[{}]
   G^*(\jmath \omega )    \cr
   I     \cr
  \end{pmat} \le 0,
 \forall \omega {\kern 4pt} \in (-\infty, +\infty).
  \end{equation}
\noindent According to the GKYP lemma (Lemma 2.3), there exists a positive symmetrical matrix $  {\mathbf P}_{m\rho 1}$ and $ {\mathbf L}_{m\rho 1},  {\mathbf K}_{m\rho 1}$ satisfying
\begin{spacing}{1.2}
\begin{subequations}\label{KYPlurC1}
   \begin{align}
&   {\mathbf A}_{m\rho 1}  {\mathbf P}_{m\rho 1} +  {\mathbf P}_{m\rho 1}  {\mathbf A}^*_{m\rho 1}       +   {\mathbf B}_{m\rho 1} {\mathbf B}^*_{m\rho 1}=-  {\mathbf L}_{m\rho 1}  {\mathbf L}_{m\rho 1}^*,   \\
&     {\mathbf P}_{m\rho 1}  {\mathbf C}^*_{m\rho 1} +   {\mathbf B}_{m\rho 1} {\mathbf D}^*_{m\rho 1} =-{\mathbf L}_{m\rho 1}  {\mathbf K}_{m\rho 1}^* \\
&   {\mathbf D}_{m\rho 1} {\mathbf D}^*_{m\rho 1}-  \gamma^2_{m\rho 1}I=-  {\mathbf K}_{m\rho 1}  {\mathbf K}_{m\rho 1}^*,
   \end{align}
\end{subequations}
\end{spacing}

\noindent Define $Q=  {\mathbf P}_{m\rho 1}, P=\rho   {\mathbf P}_{m\rho 1}$, from the above equation (\ref{KYPlurC1}) we have

\begin{equation}
\begin{array}{l}
             {\kern 10pt}-\mathbf{He}((\jmath \varpi_1 I-  A)  Q  (\jmath \varpi_2 I- A))    +  AP +  PA^* + BB^*   \\
            =  -(\jmath \varpi_c I-  A)   {\mathbf P}_{m\rho 1}  (\jmath \varpi_c I-    A )^*        + \varpi_d^2   {\mathbf P}_{m\rho 1}  -\rho  ( \jmath \varpi_1 I-A)  {\mathbf P}_{m\rho 1}  -   {\mathbf P}_{m\rho 1} ( \jmath \varpi_1 I-A)^*   + B B^*   \\
            = -(\jmath \varpi_1 I-  A)   {\mathbf P}_{m\rho 1}  (\jmath \varpi_1 I-    A )^*          -\rho  ( \jmath \varpi_1 I-A)  {\mathbf P}_{m\rho 1}  -   {\mathbf P}_{m\rho 1} ( \jmath \varpi_1 I-A)^*\\
            {\kern 8pt}+(\jmath \varpi_c I-    A ){\mathbf P}_{m\rho 1}(\jmath \varpi_d)^*+  (\jmath \varpi_d){\mathbf P}_{m\rho 1}(\jmath \varpi_c I-    A )^*+ (\jmath \varpi_d) (\jmath \varpi_d)^*{\mathbf P}_{m\rho 1} + \varpi_d^2   {\mathbf P}_{m\rho 1} + B B^*   \\
              = -(\jmath \varpi_1 I-  A)   {\mathbf P}_{m\rho 1}  (\jmath \varpi_1 I-    A )^*          -\rho  ( \jmath \varpi_1 I-A)  {\mathbf P}_{m\rho 1}  -   {\mathbf P}_{m\rho 1} ( \jmath \varpi_1 I-A)^*\\
            {\kern 8pt}+(\jmath \varpi_1 I-    A ){\mathbf P}_{m\rho 1}(\jmath \varpi_d)^*+  (\jmath \varpi_d){\mathbf P}_{m\rho 1}(\jmath \varpi_1 I-    A )^* \\
            {\kern 8pt}-(\jmath \varpi_d )(\jmath \varpi_d)^* {\mathbf P}_{m\rho 1}-  (\jmath \varpi_d)(\jmath \varpi_d)^* {\mathbf P}_{m\rho 1} + (\jmath \varpi_d) (\jmath \varpi_d)^*{\mathbf P}_{m\rho 1} + \varpi_d^2   {\mathbf P}_{m\rho 1}+ B B^*   \\
               = -\mathbf{He}((\jmath \varpi_1 I-  A)   { \mathbf P}_{m\rho 1} \left(\rho I -\jmath \varpi_d I+ 0.5* (\jmath \varpi_1 I-    A ) \right)^* )  + B B^* \\
               =   (\jmath \varpi_1 I-  A)\left\{   {\mathbf A}_{m\rho 1}{\mathbf P}_{m\rho 1}  -        {\mathbf P}_{m\rho 1}  {\mathbf A}_{m\rho 1}^*    +  {\mathbf B}_{m\rho 1} {\mathbf B}^*_{m\rho 1}  \right\}(\jmath \varpi_1 I-  A)^*    \\
            \mathop {=}\limits^{\ref{KYPlurC1}a}   (\jmath \varpi_1 I-  A) \left\{ - {\mathbf L}_{m\rho 1} {\mathbf L}_{m\rho 1}^* \right\}(\jmath \varpi_1 I-  A)^*    \\
                    \end{array}
\end{equation}

\begin{equation}
\begin{array}{l}
             {\kern 10pt} (\jmath \varpi_c I-  A)  Q C^*    +     PC^* + BD^*   \\
             =(\jmath \varpi_c I-  A)   {\mathbf P}_{m\rho 1} C^*    +     \rho  {\mathbf P}_{m\rho 1}C^* + BD^* \\
             =(\jmath \varpi_2 I-  A)   {\mathbf P}_{m\rho 1} C^*    +     (\rho+\jmath \varpi_d)  {\mathbf P}_{m\rho 1}C^* + BD^* \\
             \mathop{=}\limits^{\ref{KYPlurC1}a} (\jmath \varpi_2 I-  A)\left\{ \begin{array}{l}
              + (\jmath \varpi_2 I-  A)^{-1}BB^*(\jmath \varpi_1 I-  A)^{-*}     \\
-(\rho+\jmath \varpi_d)(\jmath \varpi_2 I-  A)^{-1}{\mathbf P}_{m\rho 1}    \\
-(\rho+\jmath \varpi_d){\mathbf P}_{m\rho 1}(\jmath \varpi_1 I-  A)^{-*}    \\
-(2\jmath \varpi_d)(\rho+\jmath \varpi_d)(\jmath \varpi_2 I-  A)^{-1} {\mathbf P}_{m\rho 1}(\jmath \varpi_1 I-  A)^{-*}  \\
 +  (\jmath \varpi_1 I-  A) (\jmath \varpi_2 I-  A)^{-1}  {\mathbf L}_{m\rho 1} {\mathbf L}_{m\rho 1}^* \\
 \end{array} \right\} C^* + BD^*  \\
 +    (\rho+\jmath \varpi_d)  {\mathbf P}_{m\rho 1}C^*\\
= -(\rho+\jmath \varpi_d) (\jmath \varpi_1 I-  A) \left\{ \begin{array}{l}   {\mathbf P}_{m\rho 1}(\jmath \varpi_1 I-  A)^{-*}C^* \\
                   -(\rho+\jmath \varpi_d)^{-1} (\jmath \varpi_1 I-  A)^{-1}  B B^*(\jmath \varpi_c I-  A)^{-*}C^* \\
                  -(\rho+\jmath \varpi_d)^{-1}(\jmath \varpi_1 I-  A)^{-1}  B D^*)\\
                                  \end{array} \right\} \\
 +  (\jmath \varpi_1 I-  A)   {\mathbf L}_{m\rho 1}  {\mathbf L}_{m\rho 1}^* C^*\\
 \mathop {=}\limits^{\ref{KYPlurC1}b} -(\rho+\jmath \varpi_d) (\jmath \varpi_1 I-  A)   \left\{     {\mathbf P}_{m\rho 1} {\mathbf C}^*_{m\rho 1}
  +   {\mathbf B}_{m\rho 1} {\mathbf D}^*_{m\rho 1} \right\}  +   (\jmath \varpi_1 I-  A)  {\mathbf L}_{m\rho 1}  {\mathbf L}_{m\rho 1}^*C^*  \\
  =- (\jmath \varpi_1 I-  A)  {\mathbf L}_{m\rho 1}\left( -(\rho-\jmath \varpi_d)    {\mathbf K}_{m\rho 1}-C   {\mathbf L}_{m\rho 1}   \right)^*\\
         \end{array}
\end{equation}

\begin{equation}
\begin{array}{l}
             {\kern 10pt}- C  Q C^*    +     DD^* -(\rho^2+\varpi_d^2)\gamma_{l\rho1}^2I   \\
             =- C  {\mathbf P}_{m\rho 1} C^*    +     DD^* -(\rho^2+\varpi_d^2)\gamma_{l\rho1}^2I   \\
       \mathop{=}\limits^{\ref{KYPlurC1}a} -C \left\{ \begin{array}{l}
              + (\jmath \varpi_1 I-  A)^{-1}BB^*(\jmath \varpi_1 I-  A)^{-*}     \\
-(\rho-\jmath \varpi_d)(\jmath \varpi_1 I-  A)^{-1}{\mathbf P}_{m\rho 1}    \\
-{\mathbf P}_{m\rho 1}(\jmath \varpi_1 I-  A)^{-*}(\rho-\jmath \varpi_d)^*    \\
 +  {\mathbf L}_{m\rho 1} {\mathbf L}_{m\rho 1}^* \\
 \end{array} \right\} C^* + DD^*-\gamma^2I  \\
    \mathop{=}\limits^{\ref{KYPlurC1}b}-         C (\jmath \varpi_1 I-  A)^{-1}BB^*(\jmath \varpi_1 I-  A)^{-*} C^*    \\
+ C(\jmath \varpi_1 I-  A)^{-1}BB^*(\jmath \varpi_1 I-  A)^{-*}C^*  +DB^*(\jmath \varpi_1 I-  A)^{-*}C^* - (\rho-\jmath \varpi_d){\mathbf K}_{m\rho 1}  {\mathbf L}_{m\rho 1}^*C^* \\
+ C(\jmath \varpi_1 I-  A)^{-1}BB^*(\jmath \varpi_1 I-  A)^{-*}C^*+C(\jmath \varpi_1 I-  A)^{-1}BD^* - C{\mathbf L}_{m\rho 1}{\mathbf K}_{m\rho 1}^*(\rho-\jmath \varpi_d)^*  \\
-C {\mathbf L}_{m\rho 1}  {\mathbf L}_{m\rho 1}^*C^* \\
  + DD^*-(\rho^2+\varpi_d^2)\gamma_{l\rho1}^2I  \\
      = (\rho^2+\varpi_d^2)
       \left\{ \begin{array}{l}
        C(\rho-\jmath \varpi_d)^{-1}(\jmath \varpi_1 I-  A)^{-1}BB^*(\jmath \varpi_1 I-  A)^{-*}(\rho-\jmath \varpi_d)^{-*}C^*\\
         +(\rho-\jmath \varpi_d)^{-1}DB^*(\jmath \varpi_1 I-  A)^{-*}(\rho-\jmath \varpi_d)^{-*}C^* \\
         +C(\jmath \varpi_1 I-  A)^{-1}(\rho-\jmath \varpi_d)^{-1}BD^*(\rho-\jmath \varpi_d)^{-*}\\
         + (\rho-\jmath \varpi_d)^{-1}DD^*(\rho-\jmath \varpi_d)^{-*}\\
         -\gamma_{l\rho1}^2I\\
      \end{array} \right\} \\
- (\rho-\jmath \varpi_d){\mathbf K}_{m\rho 1}  {\mathbf L}_{m\rho 1}^*C^*- C{\mathbf L}_{m\rho 1}{\mathbf K}_{m\rho 1}^*(\rho-\jmath \varpi_d)^*  -C {\mathbf L}_{m\rho 1}  {\mathbf L}_{m\rho 1}^*C^* \\
 \mathop {=}\limits^{\ref{PFDML1}} (\rho^2+\varpi_d^2)   \left\{     {\mathbf D}_{m\rho 1} {\mathbf D}^*_{m\rho 1}-\gamma_{l\rho1}^2I \right\}  - (\rho-\jmath \varpi_d){\mathbf K}_{m\rho 1}  {\mathbf L}_{m\rho 1}^*C^* \\
 {\kern 12pt}- C{\mathbf L}_{m\rho 1}{\mathbf K}_{m\rho 1}^*(\rho-\jmath \varpi_d)^*  -C {\mathbf L}_{m\rho 1}  {\mathbf L}_{m\rho 1}^*C^*  \\
 \mathop {=}\limits^{\ref{KYPlurC1}c}- \left(- (\rho-\jmath \varpi_d)    {\mathbf K}_{m\rho 1}-C   {\mathbf L}_{m\rho 1}   \right)\left( -(\rho-\jmath \varpi_d)    {\mathbf K}_{m\rho 1}-C   {\mathbf L}_{m\rho 1}   \right)^*\\
         \end{array}
\end{equation}

\noindent Combing the above equations, we have:
 \begin{equation}
\begin{array}{l}
 \pmatset{1}{0.1pt}
  \pmatset{0}{0.1pt}
  \pmatset{2}{4pt}
  \pmatset{3}{4pt}
  \pmatset{4}{4pt}
  \pmatset{5}{4pt}
  \pmatset{6}{1pt}
  {\kern 10pt}
\begin{pmat}[{.}]
   A   & I \cr
   C   & 0 \cr
  \end{pmat}\begin{pmat}[{.}]
   -Q   & +\jmath \varpi_c Q+ P  \cr
  -\jmath \varpi_c Q+P   & \varpi_1\varpi_2Q \cr
  \end{pmat} \begin{pmat}[{.}]
   A   & I \cr
   C   & 0 \cr
  \end{pmat}^* + \begin{pmat}[{.}]
   B   & 0 \cr
   D   & I \cr
  \end{pmat}\begin{pmat}[{|}]
    {\kern 4pt} I  {\kern 4pt}   & 0  \cr\-
   0   & -(\rho^2+\varpi_d^2) {\gamma}_{m\rho 1}^2I \cr
  \end{pmat} \begin{pmat}[{.}]
   B   & 0  \cr
   D   & I \cr
  \end{pmat}^*  \\
 = {\begin{pmat}[{|}]
  -\mathbf{He}((\jmath \varpi_1 I-  A)  Q  (\jmath \varpi_2 I- A))    +  AP +  PA^* + BB^*   & (\jmath \varpi_c I-  A)  Q C^*    +     PC^* + BD^*  \cr\-
      *  &  -C Q C^*    +   DD^*-(\rho^2+\varpi_d^2) {\gamma}_{m\rho 1}^2 I \cr
  \end{pmat}} \\[4mm]
=  {\begin{pmat}[{|}]
     -LL^*      & - LK^*  \cr\-
      *             & -KK^*   \cr
  \end{pmat}} \\
 \end{array}\end{equation}
\noindent  where
\[\begin{array}{l}
L=(\jmath \varpi_1 I-  A)  {\mathbf L}_{m\rho 1} \\
K= -(\rho-\jmath \varpi_d)    {\mathbf K}_{m\rho 1}-C   {\mathbf L}_{m\rho 1}  \\
 \end{array}\]

\noindent According to the GKYP lemma (Lemma 2.3), the following inequality can be concluded:
\begin{equation}
\label{lemma11}
\pmatset{1}{0.36pt}
  \pmatset{0}{0.2pt}
  \pmatset{2}{2pt}
  \pmatset{3}{2pt}
  \pmatset{4}{2pt}
  \pmatset{5}{2pt}
  \pmatset{6}{2pt}
\begin{pmat}[{}]
   G^*(j\omega )    \cr
   I     \cr
  \end{pmat}^*  \begin{pmat}[{|}]
    {\kern 4pt} I  {\kern 4pt}   &  0 \cr\-
   0     &  -(\rho^2+\varpi_d^2) {\gamma}_{m\rho 1}^2 \cr
  \end{pmat}  \begin{pmat}[{}]
   G^*(j\omega )    \cr
   I     \cr
  \end{pmat} \le 0,
 holds {\kern 4pt} for {\kern 4pt} all{\kern 4pt} \omega {\kern 4pt} \in [\varpi_1,\varpi_2]. \end{equation}

\noindent This leads to
\begin{equation}
\sigma_{\max}\left(G(\jmath\omega)\right) \leq (\rho^2+\varpi_d^2)^{\frac{1}{2}}   {\gamma}_{m\rho 1},   \forall \omega  \in  \Omega_m: [\varpi_1,\varpi_2].
\end{equation}

\noindent d). Since $\sigma_{\max}\left(G_{m \rho 2}({ \jmath \omega})\right) \leq   \gamma_{l \rho 2}, \forall {\omega \in \Omega: [-\infty, +\infty]}$ equivalent to\begin{equation}
\label{lemma11}
\pmatset{1}{0.36pt}
  \pmatset{0}{0.2pt}
  \pmatset{2}{2pt}
  \pmatset{3}{2pt}
  \pmatset{4}{2pt}
  \pmatset{5}{2pt}
  \pmatset{6}{2pt}
\begin{pmat}[{}]
   G^*(\jmath \omega )    \cr
   I     \cr
  \end{pmat}^*  \Pi \begin{pmat}[{}]
   G^*(\jmath \omega )    \cr
   I     \cr
  \end{pmat} \le 0,
 \forall \omega {\kern 4pt} \in (-\infty, +\infty).
  \end{equation}

\noindent According to the Continuous-time KYP lemma (Lemma 2.1), there exists a positive symmetrical matrix $  {\mathbf P}_{m\rho 2}$ and $ {\mathbf L}_{m\rho 2},  {\mathbf K}_{m\rho 12}$ satisfying
\begin{spacing}{1.2}
\begin{subequations}\label{KYPlur'e-D}
   \begin{align}
&   {\mathbf A}_{m\rho 2}  {\mathbf P}_{m\rho 2} +  {\mathbf P}_{m\rho 2}  {\mathbf A}^*_{m\rho 2}       +   {\mathbf B}_{m\rho 2} {\mathbf B}^*_{m\rho 2}=-  {\mathbf L}_{m\rho 2}  {\mathbf L}_{m\rho 2}^*,   \\
&     {\mathbf P}_{m\rho 2}  {\mathbf C}^*_{m\rho 2} +   {\mathbf B}_{m\rho 2} {\mathbf D}^*_{m\rho 2} =-{\mathbf L}_{m\rho 2}  {\mathbf K}_{m\rho 2}^* \\
&   {\mathbf D}_{m\rho 2} {\mathbf D}^*_{m\rho 2}-  \gamma^2_{m\rho 2}I=-  {\mathbf K}_{m\rho 2}  {\mathbf K}_{m\rho 2}^*,
   \end{align}
\end{subequations}
\end{spacing}

\noindent Define $Q=  {\mathbf P}_{m\rho 2}, P=\rho   {\mathbf P}_{m\rho 2}$, from the above equation (\ref{KYPlur'e-D}) and follow the similar way of the proof of statement (3), we have
\begin{equation}
\begin{array}{l}
 \pmatset{1}{0.1pt}
  \pmatset{0}{0.1pt}
  \pmatset{2}{4pt}
  \pmatset{3}{4pt}
  \pmatset{4}{4pt}
  \pmatset{5}{4pt}
  \pmatset{6}{1pt}
  {\kern 10pt}
\begin{pmat}[{.}]
   A   & I \cr
   C   & 0 \cr
  \end{pmat}\begin{pmat}[{.}]
   -Q   & +\jmath \varpi_c Q+ P  \cr
  -\jmath \varpi_c Q+P   & \varpi_1\varpi_2Q \cr
  \end{pmat} \begin{pmat}[{.}]
   A   & I \cr
   C   & 0 \cr
  \end{pmat}^* + \begin{pmat}[{.}]
   B   & 0 \cr
   D   & I \cr
  \end{pmat}\begin{pmat}[{|}]
    {\kern 4pt} I  {\kern 4pt}   & 0  \cr\-
   0   & -(\rho^2+\varpi_d^2) {\gamma}_{m\rho 1}^2I \cr
  \end{pmat} \begin{pmat}[{.}]
   B   & 0  \cr
   D   & I \cr
  \end{pmat}^*  \\
=  {\begin{pmat}[{|}]
     -LL^*      & - LK^*  \cr\-
      *             & -KK^*   \cr
  \end{pmat}} \\
 \end{array}\end{equation}
\noindent  where
\[\begin{array}{l}
L=(\jmath \varpi_2 I-  A)  {\mathbf L}_{m\rho 2} \\
K= -(\rho+\jmath \varpi_d)    {\mathbf K}_{m\rho 2}-C   {\mathbf L}_{m\rho 2}  \\
 \end{array}\]

\noindent According to the GKYP lemma (Lemma 2.3), the following inequality can be concluded:

\begin{equation}
\label{lemma11}
\pmatset{1}{0.36pt}
  \pmatset{0}{0.2pt}
  \pmatset{2}{2pt}
  \pmatset{3}{2pt}
  \pmatset{4}{2pt}
  \pmatset{5}{2pt}
  \pmatset{6}{2pt}
\begin{pmat}[{}]
   G^*(j\omega )    \cr
   I     \cr
  \end{pmat}^*  \begin{pmat}[{|}]
    {\kern 4pt} I  {\kern 4pt}   &  0 \cr\-
   0     &  -(\rho^2+\varpi_d^2) {\gamma}_{m\rho 2}^2 \cr
  \end{pmat}  \begin{pmat}[{}]
   G^*(j\omega )    \cr
   I     \cr
  \end{pmat} \le 0, \forall \omega   \in \Omega_m: [\varpi_1,\varpi_2]. \end{equation}

\noindent This leads to

\begin{equation}
\sigma_{\max}\left(G(\jmath\omega)\right) \leq (\rho^2+\varpi_d^2)^{\frac{1}{2}}   {\gamma}_{m\rho 2},   \forall \omega  \in  \Omega_m: [\varpi_1,\varpi_2].
\end{equation}

\begin{remark} \label{remark:2.2}   The linear matrix inequality of GKYP lemma (in particular, the generalized bounded real lemma) is a necessary and sufficient criteria for checking the finite-frequency maximum singular value. In contrast, the PFD bounded real lemma only provides a conservative estimation of the maximum singular value over the specified frequency range. However, the PFD bounded real lemma make it feasible to analysis the finite-frequency maximum singular value via the standard KYP Lemma (in particular, the standard bounded real lemma), in which a simpler linear matrix inequality requiring less matrix decision variables is involved. Moreover, the PFD bounded real lemma pave a way to solve some finite-frequency problems (such as the FF-MOR) by exploiting some existing entire-frequency techniques.\\
\end{remark}

\begin{remark} \label{remark:2.2} It should be noticed that the parameter matrices of all kinds of PFD mapped systems generally will be complex matrices for the general MF cases (i.e.  $\varpi_c\neq 0$). For the LF cases (i.e. $\varpi_c =0$),  the parameter matrices of the upper and lower type discrete-time PFD mapped systems are real if the parameter matrices of the given system $G_(\jmath\omega)$ are real.\\
\end{remark}

\subsection{PFD mapped systems and PFD Bounded Real Lemma (HF Case)}

\begin{definition}[PFD Mapped Systems (HF Cases)] \label{def:2.2} Let $(A,B,C,D)$ be a realization of the LTI system \eqref{originalsystem}, $\rho\in\mathbb R$, and $\varpi_h$ defining the considered high-frequency range as in Table~\ref{t9}. Then we define the following PFD mapped systems corresponding to \eqref{originalsystem}.\\
 \noindent a) \emph{{Discrete-time }system $
\hat G_{h \rho c} (e^{\jmath \theta})=:\pmatset{1}{0.36pt}
  \pmatset{0}{0.1pt}
 \pmatset{2}{1pt}
  \pmatset{3}{5pt}
  \pmatset{4}{5pt}
  \pmatset{5}{5pt}
  \pmatset{6}{5pt}\begin{pmat}[{|}]
  \hat{\mathbf A}_{h\rho c} &  \hat{\mathbf B}_{h\rho c} \cr\-
 \hat {\mathbf C}_{h\rho c} &  \hat{\mathbf D}_{h\rho c}\cr
  \end{pmat}$ constructed via the following {upper type} PFDCM $(\hat{\mathbf A}_{h\rho c},\hat{\mathbf B}_{h\rho c},\hat{\mathbf C}_{h\rho c},\hat{\mathbf D}_{h\rho c})=\hat {\mathscr M}_{h\rho c} \left(A,B,C,D,\Omega_h\right)$:}
\[\left\{ \begin{array}{l}
   \hat{\mathbf A}_{h\rho c}=(\rho^2+\varpi_h^2)^{-\frac{1}{2}}(\rho I +A) \\
   \hat{\mathbf B}_{h\rho c}=(\rho^2+\varpi_h^2)^{-\frac{1}{2}} B \\
   \hat {\mathbf C}_{h\rho c}= (\rho^2+\varpi_h^2)^{-\frac{1}{2}}   C   \\
   \hat{\mathbf D}_{h\rho c}=(\rho^2+\varpi_h^2)^{-\frac{1}{2}} D\\
 \end{array} \right.\]
\noindent \emph{will be referred as the {upper type} PFDCM system with respect to the HF range $\Omega_h$. }\\[-3mm]

\noindent b) \emph{{Discrete-time }system $\check G_{h \rho c} (e^{\jmath \theta})=:\pmatset{1}{0.36pt}
  \pmatset{0}{0.1pt}
 \pmatset{2}{1pt}
  \pmatset{3}{5pt}
  \pmatset{4}{5pt}
  \pmatset{5}{5pt}
  \pmatset{6}{5pt}\begin{pmat}[{|}]
  \check{\mathbf A}_{h\rho c} &  \check{\mathbf B}_{h\rho c} \cr\-
  \check{\mathbf C}_{h\rho c} &  \check{\mathbf D}_{h\rho c}\cr
  \end{pmat}$ constructed via the following {lower type} PFDCM $(\check{\mathbf A}_{h\rho c},\check{\mathbf B}_{h\rho c},\check{\mathbf C}_{h\rho c},\check{\mathbf D}_{h\rho c})=\check {\mathscr M}_{h\rho c} \left(A,B,C,D,\Omega_h\right)$:}
\begin{equation}\check G_{h \rho c} (e^{\jmath \theta})=:\left\{ \begin{array}{l}
   \check{\mathbf A}_{h\rho c}=   (\rho^2+1)^{\frac{1}{2}}A(\varpi_h I -\rho A)^{-1} \\
   \check{\mathbf B}_{h\rho c}=    (\varpi_h I -\rho A)^{-1} B \\
   \check{\mathbf C}_{h\rho c}=    C (\varpi_h I -\rho A)^{-1}   \\
   \check{\mathbf D}_{h\rho c}=   (\rho^2+1)^{-\frac{1}{2}} \varpi_h^{-1}\rho C(  \varpi_h I -\rho A)^{-1} B+ (\rho^2+1)^{-\frac{1}{2}} \varpi_h^{-1}   D\\
 \end{array} \right.\end{equation}
\noindent \emph{will be referred as the {lower type} PFDCM system with respect to the HF range $\Omega_h$. }\\[-3mm]

\noindent c) \emph{The following {continuous-time }system $G_{l \rho h} (\jmath \omega):\pmatset{1}{0.36pt}
  \pmatset{0}{0.1pt}
 \pmatset{2}{1pt}
  \pmatset{3}{5pt}
  \pmatset{4}{5pt}
  \pmatset{5}{5pt}
  \pmatset{6}{5pt}\begin{pmat}[{|}]
  {\mathbf A}_{h\rho1} &  {\mathbf B}_{h\rho1} \cr\-
  {\mathbf C}_{h\rho1} &  {\mathbf D}_{h\rho1}\cr
  \end{pmat}$  constructed via the following {left type} PFDCM $({\mathbf A}_{h\rho 1},{\mathbf B}_{h\rho 1},{\mathbf C}_{h\rho 1},{\mathbf D}_{h\rho 1})={\mathscr M}_{h\rho 1} \left(A,B,C,D,\Omega_h\right)$:}
\begin{equation} G_{h \rho 1} (\jmath \omega)=:\left\{ \begin{array}{l}
   {\mathbf A}_{h\rho 1}=   -0.5I+ (\rho+\jmath \varpi_h)(\jmath \varpi_h I-A)^{-1} \\
   {\mathbf B}_{h\rho 1}=    (\jmath \varpi_h I -A)^{-1} B \\
   {\mathbf C}_{h\rho 1}=    C ( \jmath \varpi_h I -A)^{-1}   \\
   {\mathbf D}_{h\rho 1}=   -(\rho-\jmath \varpi_h )^{-1}   C (\jmath \varpi_h I-A)^{-1} B -(\rho-\jmath \varpi_h )^{-1} D\\
 \end{array} \right.\end{equation}
\noindent \emph{will be referred as the {left type} PFDCM system with respect to the HF range $\Omega_h$. }\\[-3mm]

\noindent d) \emph{The following {continuous-time }system $G_{h \rho 2} (\jmath \omega):\pmatset{1}{0.36pt}
  \pmatset{0}{0.1pt}
 \pmatset{2}{1pt}
  \pmatset{3}{5pt}
  \pmatset{4}{5pt}
  \pmatset{5}{5pt}
  \pmatset{6}{5pt}\begin{pmat}[{|}]
  {\mathbf A}_{h\rho2} &  {\mathbf B}_{h\rho2} \cr\-
  {\mathbf C}_{h\rho2} &  {\mathbf D}_{h\rho2}\cr
  \end{pmat} $ constructed via the following {right type} PFDCM $({\mathbf A}_{h\rho 2},{\mathbf B}_{h\rho 2},{\mathbf C}_{h\rho 2},{\mathbf D}_{h\rho 2})={\mathscr M}_{h\rho 2} \left(A,B,C,D,\Omega_h\right)$:}
\begin{equation}  G_{h \rho 2} (\jmath \omega)=:\left\{ \begin{array}{l}
   {\mathbf A}_{h\rho 2}=   -0.5I+ (\rho-\jmath \varpi_h)(-\jmath \varpi_h I-A)^{-1} \\
   {\mathbf B}_{h\rho 2}=    (-\jmath \varpi_h I -A)^{-1} B \\
   {\mathbf C}_{h\rho 2}=    C ( -\jmath \varpi_h I -A)^{-1}   \\
   {\mathbf D}_{h\rho 2}=   (\rho-\jmath \varpi_h )^{-1}   C (-\jmath \varpi_h I-A)^{-1} B +(\rho-\jmath \varpi_h )^{-1} D\\
 \end{array} \right.\end{equation}
\noindent \emph{will be referred as the {right type} PFDCM system with respect to the HF range $\Omega_h$. }\\[-3mm]
\end{definition}

\begin{proposition}\label{prop:2.2}
Let
$\rho_h^*=\max \left\{ \frac{\varpi_h^2-\mathfrak{Re}(\lambda_i)^2-(\mathfrak{Im}(\lambda_i))^2}{2\mathfrak{Re}(\lambda_i)}\,|\, i=1,\ldots,n\right\}$, where $\lambda_i$, $i=1,\ldots,n$, are the eigenvalues of $A$, then the following statements hold.
\begin{enumerate}[a)]
\item
If $\rho>\rho_h^*$,  then the matrix  $\hat{\mathbf A}_{h\rho c}$ is Schur stable.
\item
If $\rho<-\rho_h^*$,  then the matrix $\check{\mathbf A}_{h\rho c}$ is Schur stable.
\item
If $\rho>\rho_h^*$,  then the matrix  ${\mathbf A}_{h\rho 1}$ is Hurwitz stable.
\item
If $\rho>\rho_h^*$,   then the matrix  ${\mathbf A}_{h\rho 2}$ is Hurwitz stable.
\end{enumerate}
\end{proposition}
\begin{proof}
The proof is analogous to the proof of Proposition~\ref{prop:2.1} and is omitted here.
\end{proof}

\begin{theorem}\label{thm:2.2}
The following statements on the relationship between the maximum singular value of the mapped systems and the given system hold:
\begin{enumerate}[a)]
\item
If $\sigma_{\max}(\hat G_{h \rho c}({e^{\jmath \theta}})) \leq \hat \gamma_{h \rho c} \; \forall {\theta \in \Theta}$,  then $\sigma_{\max}\left(G(\jmath \omega)\right) \leq (\rho^2+\varpi_h^2)^{\frac{1}{2}} \hat \gamma_{h \rho c}\; \forall \omega \in \Omega_h$.
\item
If $\sigma_{\max}(\check G_{h \rho c}({e^{\jmath \theta}})) \leq \check \gamma_{h \rho c}\; \forall {\theta \in \Theta}$,  then $\sigma_{\max}\left(G(\jmath \omega)\right) \leq  (\rho^2+\varpi_h^2)^{\frac{1}{2}} \check \gamma_{h \rho c}\; \forall \omega \in \Omega_h$.
\item
If $\sigma_{\max}\left(G_{h \rho 1}({\jmath \omega} )\right) \leq \gamma_{h \rho 1}\; \forall {\omega \in \Omega}$, then $\sigma_{\max}\left(G(\jmath \omega)\right) \leq   (\rho^2+\varpi_h^2)^{\frac{1}{2}} \gamma_{h \rho 1}\; \forall \omega \in \Omega_h$.
\item
If $\sigma_{\max}\left(G_{h \rho 2}({\jmath \omega})\right) \leq \gamma_{h \rho 2}\; \forall{\omega \in \Omega}$, then $\sigma_{\max}\left(G(\jmath \omega)\right) \leq  (\rho^2+\varpi_h^2)^{\frac{1}{2}}  \gamma_{h \rho 2}\; \forall \omega \in \Omega_h$.
\end{enumerate}
\end{theorem}
\begin{proof}
The proof is similar to the proof of Theorem~\ref{thm:2.1} and is therefore omitted.
\end{proof}

\section{Parameterized Frequency-dependent Balanced Truncation}

In this section, we first summarize the results of the standard LyaBT in the discrete-time setting in subsection \uppercase\expandafter{\romannumeral3}.A. Afterwards, the results on the new proposed PFDBT schemes for LF cases and HF cases are presented, respectively.

\subsection{Review of the standard LyaBT}

\begin{spacing}{1.2}
\begin{algorithm}
\caption{Continuous-time (discrete-time) standard LyaBT}
\begin{algorithmic}
\REQUIRE {Full-order continuous-time system $G(\jmath \omega):(A,B,C,D)$ or discrete-time system $G(e^{\jmath \theta}):(A,B,C,D)$, and the order of reduced model $r$}, \\[3mm]
  \textbf{Step 1.} For continuous-time case, solve the continuous-time  controllability and observability Lyapunov equations
  \begin{subequations}\label{Con-LyaEqu}
   \begin{align}
   &  A     P ^c    - P ^cA^*     +    B     B^* =0, \\[-1mm]
   &  A^*   P ^o    -   P ^o A      +     C ^*  C =0,
   \end{align}
   \end{subequations}
 For discrete-time case, solve the continuous-time  controllability and observability Lyapunov equations
  \begin{subequations}\label{dis-LyaEqu}
   \begin{align}
   &  A     P ^c A^*    - P ^c     +    B     B^* =0, \\[-1mm]
   &  A^*   P ^o A    -   P ^o       +     C ^*  C =0,
   \end{align}
\end{subequations}
  \textbf{Step 2.} Compute the Cholesky factorization                $ P^c= UU$. \\[3mm]
  \textbf{Step 3.} Compute the eigenvalue decomposition of $U_{l\rho}^*Q_{l\rho}^o$, i.e., $ U^*P^oU= V\Sigma^2 V^*$. \\[3mm]
  \textbf{Step 4.} Compute the coordinate transformation matrix:              $ T= \Sigma^{\frac{1}{2}}V^*U^{-1}$ \\[3mm]
  \textbf{Step 5.} Compute the balanced realization of the given system by coordinate transformation:
\begin{equation}
(A_{b}, B_{b}, C_{b}, D_{b}) = ( T^{-1}  A  T,  T^{-1} B,     C T,  D)
\end{equation}
\textbf{Step 6.} Compute the reduced-order model as $G_r(\jmath \omega)$
\begin{equation}
\label{truncatedsystemc}\begin{array}{l}
(A_r,B_r,C_r,D_r) =(Z_r A_b Z_r^T, Z_r B_b, C_b Z_r^T, D_b). \\
 \end{array}\end{equation}
\noindent {\kern 40pt}where $Z_r=[I_r,\mathbf 0_{(r,n-r)}]$ is the truncating matrix with respect to the reduced order $r$. \\[3mm]

\ENSURE Reduced-order model $G_r(\jmath \omega):(A_r,B_r,C_r,D_r)$
\end{algorithmic}
\end{algorithm}

\end{spacing}

\begin{lemma}(\cite{morMoo81}, \cite{morEnn84},\cite{morPerS82}, \cite{ZhoD96}) \label{lem3.1} For a given linear continuous-time system $G(\jmath\omega)$ or discrete-time system $G(e^{\jmath\theta})$, suppose the continuous-time reduced model $G_r(\jmath\omega)$ or discrete-time reduced model $G_r(e^{\jmath\theta})$ is generated via the standard LyaBT, then the following EF-type error bound holds, i.e.   \\
a).  For continuous-time case, the EF-type error bound is \begin{equation}
\label{LyaBT-EFEB-C}
{\sigma _{\max }}\left( {G(\jmath\omega)-G_r(\jmath\omega)} \right) \leq 2\sum\limits_{i = n}^{r + 1} {{\sigma _i}}, \forall \omega \in \Omega:(-\infty, +\infty)
\end{equation}
b).  For discrete-time case, the EF-type error bound is  \begin{equation}
\label{LyaBT-EFEB-D}
{\sigma _{\max }}\left( {G(e^{\jmath\theta})-G_r(e^{\jmath\theta})} \right) \leq 2\sum\limits_{i = n}^{r + 1} {{\sigma _i}}, \forall \theta \in \Theta:(-\pi, +\pi)
\end{equation}
\end{lemma}~\\

\begin{remark} \label{rem3.1} For more details on the continuous-time EF-type error bound, please refer to  \cite{morEnn84} \cite{ZhoD96} . For more details on the discrete-time EF-type error bound, please refer to \cite{morAlF87}, \cite{ZhoD96}. Besides, as the companion version of the standard LyaBT, SPA also provides the same EF-type error bounds \cite{morLiuA89}. It should be pointed out that the KYP Lemma plays a important role in the proof of EF-type error bound. One could find a KYP lemma based constructive way to prove the EF-type error bound in \cite{ZhoD96}.
\end{remark}~\\

\subsection{\textbf{PFD Balanced Truncation (LF Case)}}

\noindent Based upon the above preliminaries and results, we now at the stage to present the PFDBT algorithm for LF case.

\begin{spacing}{1.4}
\begin{algorithm}
\caption{PFDBT (LF Case)}
\begin{algorithmic}
\REQUIRE {Full-order model $(A,B,C,D)$, frequency interval $\Omega_l:[-\varpi_l, +\varpi_l]$, user-defined admissible parameter $\rho$ and the order of reduced model $(r)$}, \\[3mm]
\textbf{Routing 1.} \\
apply the standard discrete-time LyaBT for the mapped discrete-time system $\hat {\mathbf G}_{m\rho c}(e^{\jmath \theta})$ to obtain the mapped discrete-time reduced model
$\hat {\mathbf G}_{m\rho cr}(e^{\jmath \theta}): (\hat {\mathbf A}_{l\rho cr},\hat {\mathbf B}_{l\rho cr},\hat {\mathbf C}_{l\rho cr},\hat {\mathbf D}_{l\rho cr})$. Compute the reduced-order model by applying inverse upper type PFD mapping as follows:
\begin{equation}
\label{LF-PFDBT-c1}
\begin{array}{l}
 \hat A_{r}=   (\rho I+\jmath \varpi_c I)-(\rho^2+\varpi_d^2)^{\frac{1}{2}} \hat {\mathbf A}_{l\rho cr} ^{-1},\\
 \hat B_{r} =  (\rho I+\jmath \varpi_c I - \hat A_{r}) \hat {\mathbf B}_{l\rho cr}, \\
 \hat C_{r} =  \hat {\mathbf C}_{l\rho cr} (\rho I+\jmath \varpi_c I - \hat A_{r}),\\
 \hat D_{r} = \hat {\mathbf D}_{l\rho cr} -\hat C_{r}(\rho  I + \jmath \varpi_c I - \hat A_{r})^{-1}\hat B_{r}. \\
 \end{array}\end{equation}
where $\varpi_c=0$ and $\varpi_d=\varpi_l$.\\
\textbf{Routing 2. } \\
apply the standard discrete-time LyaBT for the discrete-time PFD mapped system $\hat {\mathbf G}_{m\rho c}(e^{\jmath \theta})$, obtain the discrete-time mapped reduced model
$\hat {\mathbf G}_{m\rho cr}(e^{\jmath \theta}): (\hat {\mathbf A}_{l\rho cr},\hat {\mathbf B}_{l\rho cr},\hat {\mathbf C}_{l\rho cr},\hat {\mathbf D}_{l\rho cr})$. Compute the reduced-order model by applying inverse upper type PFD mapping as follows:
 \begin{equation}
\label{LF-PFDBT-c2}\begin{array}{l}
 \check A_{r}= -\jmath \varpi_c I-   \varpi_d  (\rho^2+1)^{-\frac{1}{2}}(\rho(\rho^2+1)^{-\frac{1}{2}}-\check {\mathbf A}_{m\rho c r})^{-1},\\
 \check B_{r} = (\rho^2+1)^{\frac{1}{2}} (\jmath \varpi_c I - \check {A}_{r}) \check {\mathbf B}_{m\rho c r}, \\
 \check C_{r} = (\rho^2+1)^{\frac{1}{2}} \check {\mathbf C}_{m\rho c r}(\jmath \varpi_c I - \check {A}_{r}),\\
 \check D_{r} =  (\rho^2+1)^{\frac{1}{2}} \varpi_d \check {\mathbf D}_{l\rho cr} -\check C_{r}(\jmath \varpi_c I - \check A_{r})^{-1} \check B_{r}. \\
 \end{array}\end{equation}
where $\varpi_c=0$ and $\varpi_d=\varpi_l$.
 \ENSURE Reduced-order model: $G_r(\jmath\omega):(\hat A_{r},\hat B_{r},\hat C_{r},\hat D_{r})$ or  $G_r(\jmath\omega):(\check A_{r},\check B_{r},\check C_{r},\check D_{r})$.
\end{algorithmic}
\end{algorithm}
\end{spacing}

\begin{theorem}[LF-type error bound via LF case PFDBT]\label{thm:4.2} Given a linear continuous-time system $G(\jmath \omega )$ and a pre-known LF interval $\omega \in \Omega_{l}: [-\varpi_l,+\varpi_l]$. Suppose the reduced model ${G_r}(\jmath \omega)$ is generated via the LF case PFDBT algorithm, then the approximation performance over pre-specified frequency interval satisfys the following FF-type error bound:
\begin{equation}\begin{array}{l}
\label{LF-errorbound}
\sigma_{max}(G(\jmath \omega ) - {G_r}(\jmath \omega))  \le 2 (\rho^2+\varpi_l^2)^{\frac{1}{2}} \sum\limits_{i = r + 1}^{n} {{\sigma _{li}}}, {\kern 6pt} \omega \in \Omega_l=:[-\varpi_l,+\varpi_l]
 \end{array}.
\end{equation}
\end{theorem}

\begin{proof} The error system between the original high-order system model $G(\jmath \omega)$ and
the truncated $(n-1)^{th}$ reduced system $G_{r}(\jmath \omega)$ can be
represented by
\begin{equation}\begin{array}{l}
\label{errorsystem}
E_r(\jmath \omega)=G(\jmath \omega)-G_r(\jmath \omega)=:\pmatset{1}{0.36pt}
  \pmatset{0}{0.2pt}
  \pmatset{2}{8pt}
  \pmatset{3}{8pt}
  \pmatset{4}{8pt}
  \pmatset{5}{4pt}
  \pmatset{6}{4pt}    \begin{pmat}[{|}]
        \mathcal A_{er}   &    \mathcal B_{er} \cr\-
        \mathcal C_{er}   &    \mathcal D_{er} \cr
      \end{pmat}    =
  \pmatset{1}{0.36pt}
  \pmatset{0}{0.2pt}
  \pmatset{2}{4pt}
  \pmatset{3}{4pt}
  \pmatset{4}{4pt}
  \pmatset{5}{2pt}
  \pmatset{6}{2pt}
 \begin{pmat}[{.|}]
      A_r   & 0  &   B_r\cr
      0     & A  &   B  \cr\-
       -C_r & C  &   D-D_r \cr
      \end{pmat}     \\
 \end{array}.\end{equation}\\
\noindent suppose the parameter matrices $(A_r,B_r,C_r,D_r)$ are computed via upper routine, then apply the upper case PFD mapping for the error system (\ref{errorsystem}). It can be concluded that the mapped error system can be represented by
\begin{equation}\begin{array}{l}
\label{mapederrorsystem}
\hat {\mathbf E}_{m\rho cr}(e^{\jmath \theta})=\hat {\mathbf G}_{m\rho c}(e^{\jmath \theta})-\hat {\mathbf G}_{m\rho c r}(e^{\jmath \theta})=:\pmatset{1}{0.36pt}
  \pmatset{0}{0.2pt}
  \pmatset{2}{8pt}
  \pmatset{3}{8pt}
  \pmatset{4}{8pt}
  \pmatset{5}{4pt}
  \pmatset{6}{4pt}    \begin{pmat}[{|}]
       \hat {\mathcal A}_{m\rho c er}   &    \hat {\mathcal B}_{m\rho c er} \cr\-
       \hat {\mathcal C}_{m\rho c er}   &    \hat {\mathcal D}_{m\rho c er}  \cr
      \end{pmat}    =
  \pmatset{1}{0.36pt}
  \pmatset{0}{0.2pt}
  \pmatset{2}{4pt}
  \pmatset{3}{4pt}
  \pmatset{4}{4pt}
  \pmatset{5}{2pt}
  \pmatset{6}{2pt}
 \begin{pmat}[{.|}]
      \hat {\mathbf A}_{m\rho c r}   & 0  &   \hat {\mathbf B}_{m\rho c r} \cr
      0     & \hat {\mathbf A}_{m\rho c}   &   \hat {\mathbf B}_{m\rho c}   \cr\-
       -\hat {\mathbf C}_{m\rho c r}  & \hat {\mathbf C}_{m\rho c}  &   \hat {\mathbf D}_{m\rho c} -\hat {\mathbf D}_{m\rho c r}  \cr
      \end{pmat}     \\
 \end{array}.\end{equation}\\
\noindent Since $\hat {\mathbf G}_{m\rho c r}(\jmath \omega)$ is the reduced model obtained by applying the standard LyaBT for the upper PFD mapped system $\hat {\mathbf G}_{m\rho c}(\jmath \omega)$. According to the Lemma 4, we have
\begin{equation}\begin{array}{l}
\label{LF-errorbound}
\sigma_{max}(\hat {\mathbf G}_{m\rho c}(e^{\jmath \theta})-\hat {\mathbf G}_{m\rho c r}(e^{\jmath \theta}))  \le 2 \sum\limits_{i = r + 1}^{n} {{\sigma _{l\rho i}}}, {\kern 6pt} \theta \in \Theta:=[-\pi, +\pi)
 \end{array}.
\end{equation}

 \noindent Noticing that the error system $\hat {\mathbf E}_{m\rho cr}(\jmath \omega)$ can be obtained by applying the upper type PFD mapping on error system $E_r(\jmath \omega)$, then we have
\begin{equation}\begin{array}{l}
\label{LF-errorbound}
\sigma_{max}(G(\jmath \omega ) - {G_r}(\jmath \omega))  \le 2 (\rho^2+\varpi_l^2)^{\frac{1}{2}} \sum\limits_{i = r + 1}^{n} {{\sigma _{l\rho i}}}, {\kern 6pt} \omega \in \Omega_l=:[-\varpi_l,+\varpi_l]
 \end{array}.
\end{equation}
\noindent according to Theorem \ref{thm:2.1}.  In the cases that the parameter matrices of the reduced model is computed via  routine2 of the PFDBT algorithm, one can prove the LF-type error bound similarly. Thus, the proof is completed. \\ \end{proof}

\begin{remark} \label{rem:3.2} As far as our knowledge, this is the first result that provides FF-type error bound in the framework of balanced truncation. Similar with the EF-type error bound (\ref{LyaBT-EFEB-C}) provided by LyaBT, the FF-type error bound (\ref{LF-errorbound}) is also very simple and \emph{a priori}. Comparing the values of EF-type error bound with FF-type error bound theoretically is difficult, however,  it is shown that the FF-type error bound could be smaller than the EF-type error bound by choosing a proper parameter $\rho$. To obtain a proper value of the parameter $\rho$, we suggest a simple line search over the admissible range of $\rho$. As shown by the examples in the sequel, one could find the proper parameter by observing the curves of FF-type error bound with respect to several different values of the parameter $\rho$. How to compute the optimal parameter rendering the FF-type error bound as the smallest value is still an open problem for further investigation.
\end{remark}

\begin{remark} \label{rem:3.2} It is well-known that the original model is required to be stable to apply the standard LyaBT, moreover, the stability will be preserved by the reduced model generated via the standard LyaBT.  The stability restriction on the original model is not needed for PFDBT. For non-stable original model, one could apply the PFDBT just by choosing a larger enough parameter rendering the PFD mapped matrices $\hat {\mathbf A}_{m\rho c}$ or $\check {\mathbf A}_{m\rho c}$ be Schur stable. At the same time, the PFDBT don't possesses the stability preservation property. In other words, the stability of reduced model cannot be theoretically guaranteed even the original model is stable. According to our numerical experiments, one could always obtain a stable reduced model in cases that the original model is stable by selecting a proper parameter (especially by letting the parameter $\rho$ large enough).
\end{remark}

\begin{remark} \label{rem:3.3} In algorithm 2, only the discrete-time PFD mapped systems and the discrete-time LyaBT procedures are involved. Obviously, if we resort to the continuous-time PFD mapped systems and the continuous-time LyaBT procedures in a similar way, another routines give rise to reduced models could be derived. Unfortunately, the parameter matrices of the reduced models generally will become complex matrices under such a circumstance. Besides, extending the PFDBT for LF case to the MF case is also feasible. Likewise, such an extension generally will leads to complex reduced models since $\varpi_c \neq 0$.
\end{remark}

\subsection{\textbf{PFD Balanced Truncation (HF Case)}}

Similarly with the LF case, we now present the HF case PFDBT algorithm and the results on HF-type error bound. \\

\begin{spacing}{1.4}
\begin{algorithm}
\caption{PFDBT(HF case)}
\begin{algorithmic}
\REQUIRE {Full-order model $(A,B,C,D)$, HF frequency range $\Omega_h: (-\infty, -\varpi_h]\cup [+\varpi_h, +\infty)$, user-defined admissible parameter $\rho$ and the order of reduced model $(r)$}, \\[3mm]
\textbf{Routing 1}. apply the standard discrete-time LyaBT for the mapped discrete-time system $\hat {\mathbf G}_{h\rho c}(e^{\jmath \theta})$ to obtain the mapped discrete-time reduced model
$\hat {\mathbf G}_{h\rho cr}(e^{\jmath \theta}): (\hat {\mathbf A}_{h\rho cr},\hat {\mathbf B}_{h\rho cr},\hat {\mathbf C}_{h\rho cr},\hat {\mathbf D}_{h\rho cr})$. Compute the reduced-order model by applying inverse upper type PFD mapping as follows:
\begin{equation}
\label{LF-PFDBT-c1}
\begin{array}{l}
 \hat A_{r}=   (\rho^2+\varpi_h^2)^{\frac{1}{2}} \hat {\mathbf A}_{h\rho cr} -\rho I,\\
 \hat B_{r} =  (\rho^2+\varpi_h^2)^{\frac{1}{2}} \hat {\mathbf B}_{h\rho cr}, \\
 \hat C_{r} =  (\rho^2+\varpi_h^2)^{\frac{1}{2}} \hat {\mathbf C}_{h\rho cr},\\
 \hat D_{r} =  (\rho^2+\varpi_h^2)^{\frac{1}{2}} \hat {\mathbf D}_{h\rho cr}. \\
 \end{array}\end{equation}

\textbf{Routing 2.} apply the standard discrete-time LyaBT for the mapped discrete-time system $\hat {\mathbf G}_{h\rho c}(e^{\jmath \theta})$ to obtain the mapped discrete-time reduced model
$\hat {\mathbf G}_{h\rho cr}(e^{\jmath \theta}): (\hat {\mathbf A}_{h\rho cr},\hat {\mathbf B}_{h\rho cr},\hat {\mathbf C}_{h\rho cr},\hat {\mathbf D}_{h\rho cr})$. Compute the reduced-order model by applying inverse upper type PFD mapping as follows:
 \begin{equation}
\label{LF-PFDBT-c2}\begin{array}{l}
 \check A_{r}=  \varpi_h(\rho^2+1)^{-\frac{1}{2}}  \check {\mathbf A}_{h\rho c r} (I+\rho(\rho^2+1)^{-\frac{1}{2}}\check {\mathbf A}_{h\rho c r})^{-1},\\
 \check B_{r} =  (\varpi_h I - \rho \check {A}_{r}) \check {\mathbf B}_{h\rho c r}, \\
 \check C_{r} =   \check {\mathbf C}_{h\rho c r}(\varpi_h I - \rho \check {A}_{r}) ,\\
 \check D_{r} =  \varpi_h(\rho^2+1)^{\frac{1}{2}} ( \check {\mathbf D}_{h\rho cr} -\rho(\rho^2+1)^{-\frac{1}{2}}\varpi_h^{-1} \check C_{r}(\varpi_h I - \rho \check A_{r})^{-1} \check B_{r}). \\
 \end{array}\end{equation}

\ENSURE Reduced-order model: $G_r(\jmath\omega):(\hat A_{r},\hat B_{r},\hat C_{r},\hat D_{r})$ or  $G_r(\jmath\omega):(\check A_{r},\check B_{r},\check C_{r},\check D_{r})$.
\end{algorithmic}
\end{algorithm}
\end{spacing}

\begin{theorem} [HF-type error bound via HF case PFDBT] \label{thm:4.3} Given a linear continuous-time system $G(\jmath\omega)$ and a pre-known HF frequency interval $\omega \in \Omega_{h}: (-\infty, -\varpi_h] \cup [+\varpi_h,+\infty)$. Suppose the reduced model $G_r(\jmath\omega)$ is generated via PFDBT, then the approximation performance over pre-specified frequency interval satisfy the following HF-type error bound:
\begin{equation}\begin{array}{l}
\label{LF-errorbound}
\sigma_{max}(G(\jmath \omega ) - {G_r}(\jmath \omega))  \le 2 (\rho^2+\varpi_h^2)^{\frac{1}{2}} \sum\limits_{i = r + 1}^{n} {{\sigma _{h\rho i}}}, {\kern 6pt} \omega \in \Omega_h=:(-\infty, -\varpi_h]\cup[+\varpi_h, +\infty)
 \end{array}.
\end{equation}
\end{theorem}

\begin{proof} The proof can be completed in a similar way of the prove of Theorem \ref{thm:4.2}
\end{proof}

\section{Illustrating Examples}

\noindent In this section we demonstrate the validity of the PFD bounded real lemmas and the advantages of the PFDBT schemes through four examples.

\begin{example} \label{exp:1} Lets consider a simple linear continuous-time system (\ref{originalsystem}) with the following parameter matrices:
\begin{equation}
\small
\begin{array}{l}
 \pmatset{1}{0.36pt}
  \pmatset{0}{0.2pt}
  \pmatset{2}{4pt}
  \pmatset{3}{4pt}
  \pmatset{4}{6pt}
  \pmatset{5}{6pt}
  \pmatset{6}{6pt}    \begin{pmat}[{|}]
       A    &   B \cr\-
       C    &   D   \cr
      \end{pmat}    =
 \pmatset{1}{0.36pt}
  \pmatset{0}{0.2pt}
  \pmatset{2}{4pt}
  \pmatset{3}{4pt}
  \pmatset{4}{6pt}
  \pmatset{5}{6pt}
  \pmatset{6}{6pt}
 \begin{pmat}[{.|}]
  -4.1859   &   0.7195  & 1.8712\cr
   1.7797   &  -1.1872  & 1.1639\cr\-
   0.4528   & -2.4099   & 2.5606 \cr
      \end{pmat}.
 \end{array}\end{equation}
\noindent We are interested to apply the proposed PFD bounded real lemma for estimating the maximum singular value of this system over four different low-frequency ranges: $\Omega_l^1:[-0.1,0.1],\Omega_l^2:[-1,1],\Omega_l^3:[-10,10],\Omega_l^4:[-100,100]$.

 \begin{figure}[ht!]
    \centering
      \includegraphics[scale=0.58]{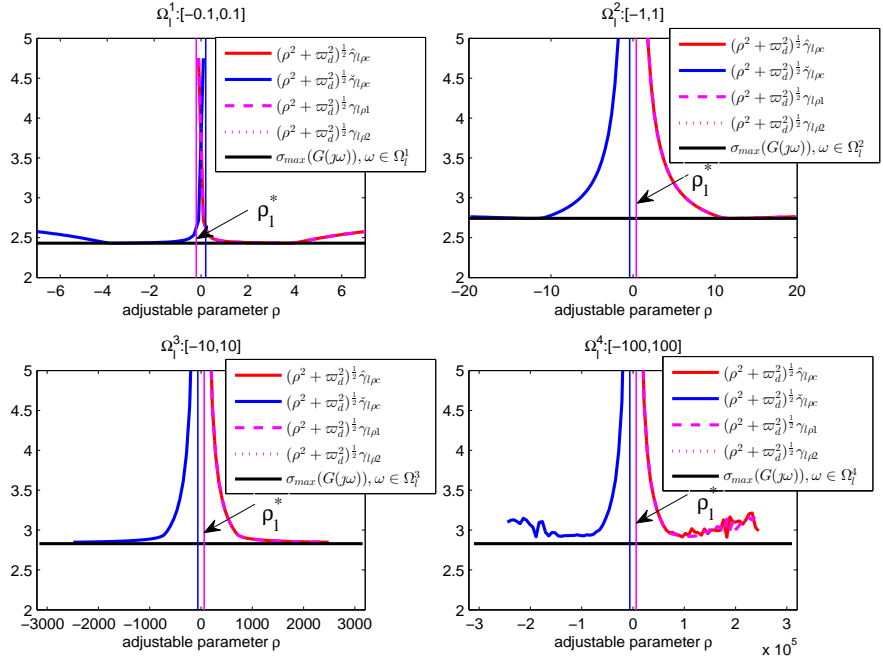}
               \caption{Estimating the maximum singular value of given system over specified frequency range via PFD bounded real lemma}
       \end{figure}

       \noindent As shown in Fig. 1, the estimated maximum singular values obtained by PFD bounded real lemma with any admissible parameter $\rho$ are always lager than the actual maximum singular values over the specified low-frequency ranges.  In particular, the gaps between the estimated maximum singular values and the actual maximum singular value may be very small if the adjustable parameter $\rho$ lies in an appropriate range. The results indicate that the validity and effectiveness of the proposed PFD bounded real lemma. \\
\end{example}

\begin{example} Lets consider a linear continuous-time system (\ref{originalsystem}) with the following parameter matrices:   \begin{equation}
\small
\begin{array}{l}
 \pmatset{1}{0.36pt}
  \pmatset{0}{0.2pt}
  \pmatset{2}{8pt}
  \pmatset{3}{8pt}
  \pmatset{4}{8pt}
  \pmatset{5}{8pt}
  \pmatset{6}{8pt}    \begin{pmat}[{|}]
       A    &   B \cr\-
       C    &   D   \cr
      \end{pmat}    =
 \pmatset{1}{0.36pt}
  \pmatset{0}{0.2pt}
  \pmatset{2}{2pt}
  \pmatset{3}{2pt}
  \pmatset{4}{2pt}
  \pmatset{5}{2pt}
  \pmatset{6}{2pt}
 \begin{pmat}[{.....|}]
   -4.7488  &  0.3264  &  1.9341  &  -1.2358 &   1.4344 &   1.0027 &  0.0971  \cr
   -0.8072  & -1.9578  &  -1.2402 &   0.4604 &  -1.3092 &   0.7351 &  -0.0346 \cr
    1.2614  & -0.9532  & -5.7282  &  1.4590  &  1.9886  & -1.7071  &  2.6406  \cr
    0.2184  & -0.8236  &  0.6495  & -4.7123  &  1.3120  &  0.2781  & -1.8819  \cr
   -1.4203  & -1.9980  & -0.6598  & -0.2915  & -3.4583  & -1.5371  & 1.9220   \cr
   -1.2009  & -1.6311  &  0.1655  & -1.3573  &  1.5405  & -3.5409  &  -0.4961 \cr\-
   1.9256   &  1.4937  & -0.4044  &  0.7905  & -0.4776  &  2.0169  & 0.9839   \cr
      \end{pmat}.
 \end{array}\end{equation}

\noindent  Consider two different frequency range $\Omega_l^1:[-1,1]$ and $\Omega_l^2:[-2,2]$. In order to show the differences between the standard LyaBT, SPA and the proposed PFDBT, the EF-type error bound via LyaBT(SPA),  the FF-type error bound via PFDBT as well as the actual approximation error are depicted by the following Fig. 2 and Fig. 3.

 \begin{figure}[ht!]
    \centering
      \includegraphics[scale=0.59]{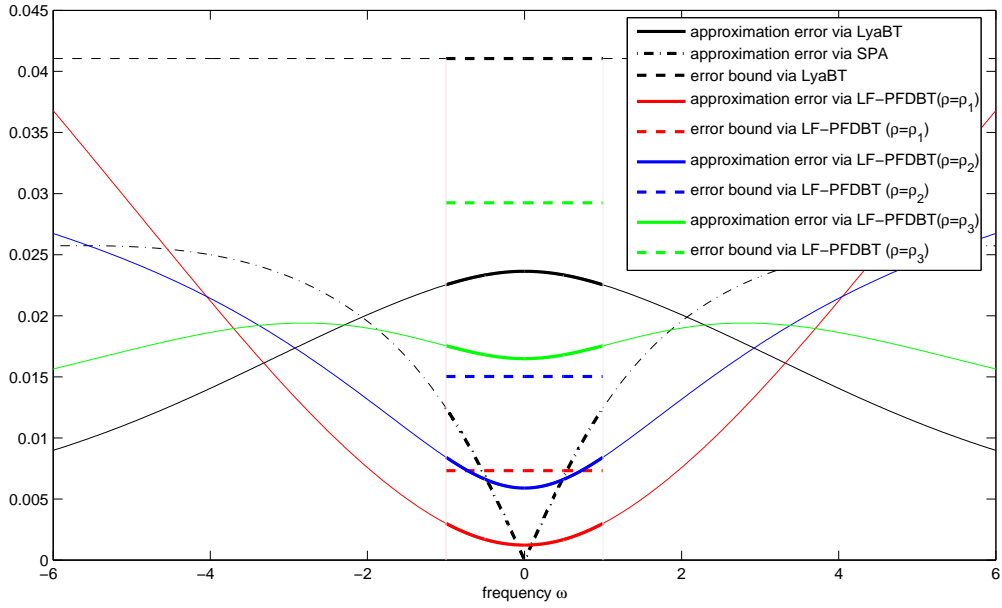}
       \caption{comparison between the standard LyaBT, SPA and the proposed PFDBT ($\Omega_l^1:[-1,1]$)}
       \end{figure}

        \begin{figure}[ht!]
    \centering
      \includegraphics[scale=0.59]{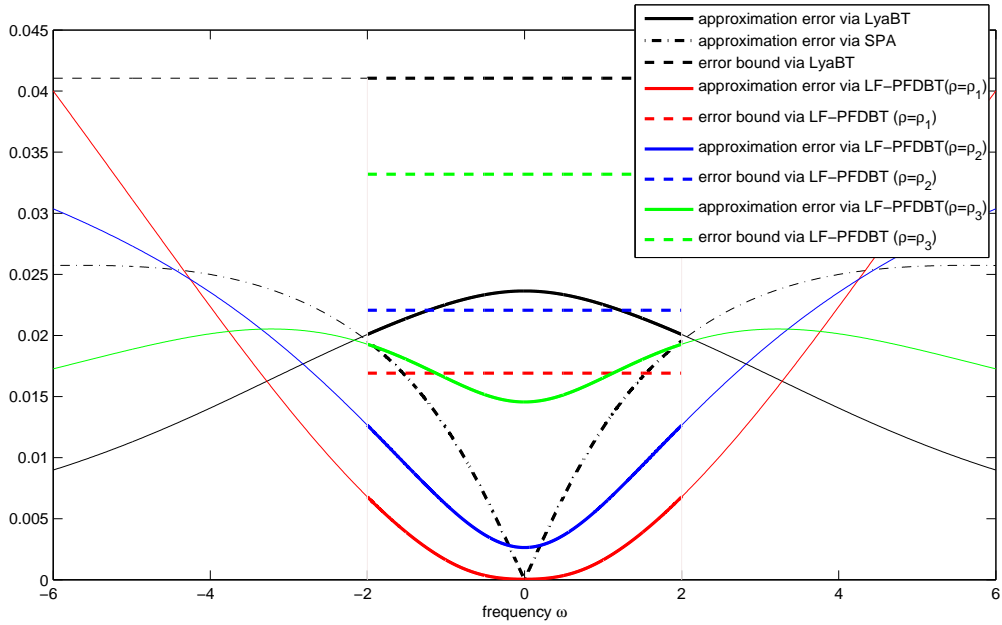}
                 \caption{comparison between the standard LyaBT, SPA, and the proposed PFDBT ($\Omega_l^2:[-2,2]$)}
       \end{figure}

\noindent To apply the proposed PFDBT, here we just randomly choose three different admissible values of the parameter $\rho$ ($\rho_1=4, \rho_2=7,\rho=20$). As Fig. 2 and Fig. 3 illustrate, the proposed PFDBT performs better than the standard LyaBT. In particular, the actual in-band approximation error resulted by PFDBT also could be smaller than the actual in-band error obtained by SPA, which is well-known as for good low-frequency approximation performance. More importantly, the PFDBT possesses an advantage on the in-band approximation error estimation. Obviously, the FF-type error bounds provided by PFDBT are smaller than the EF-type error bound provided by LyaBT(SPA). This property makes the proposed PFDBT more appealing for selecting the minimum order of the reduced model satisfying \emph{a priori} given error tolerance.\\
\end{example}

\begin{example}[The CD player benchmark example \cite{morChaP05}] \label{exp:3}
\noindent This original model of benchmark CD player example describes the dynamics between a swing arm on which a lens is mounted by means of two horizontal leaf springs.   The model has $120$ states, i.e., $n=120$ (Please refer to \cite{morChaP05} for more details). Suppose the interested frequency ranges are of low-frequency type, here we are intended to compare the achievable in-band error bound by applying the standard LyaBT and the proposed PFDBT.
 \begin{figure}[ht!]
    \centering
      \includegraphics[scale=0.58]{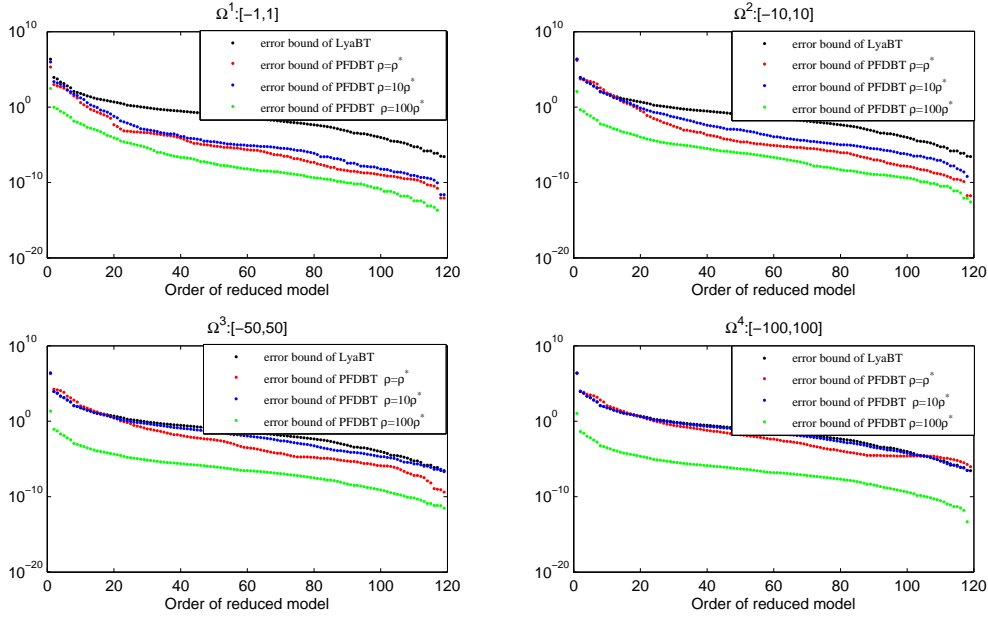}
                  \caption{comparison between the EF-type error bound via standard LyaBT and the FF-type error bound via PFDBT}
       \end{figure}

\noindent Given four different low-frequency ranges $\Omega_l^i,i=1,2,3,4$, the corresponding FF-type error bounds with different values ($\rho^{i*}, 10\rho^{i*},100\rho^{i*}, i=1,2,3,4$) of the adjustable parameter are depicted in Fig. 4, where $\rho^{i*}, i=1,2,3,4$ is the minimum value rendering the PFD mapped system $\hat G_{m \rho c} (e^{\jmath \theta})$ Schur stable.   For comparison, the EF-type error bounds obtained by standard LyaBT are also included.  From Fig. 4, it is clear that the PFDBT is possible to give rise to a smaller in-band error bound. Certainly, to what extend the in-band error bound can be improved is depended on the choice of parameter $\rho$.\\
\end{example}

\begin{example}[The ISS benchmark example \cite{morChaP05}] \label{exp:4} This is a model of component $1r$ (Russian service module) of the ISS. It has 270 states, 3 inputs and 3 outputs (Please refer to \cite{morChaP05}
for more details). Here we are interested to approximate the original model over a high-frequency $\Omega_h: (-\infty, -35] \cup [35, +\infty)$. Suppose there exists \emph{a priori} assigned error tolerance on the in-band approximation performance as follows,
\[\begin{array}{l}
\sigma_{max}(G(\jmath \omega ) - {G_r}(\jmath \omega))  \le 0.001, \omega \in \Omega_h: (-\infty, -35] \cup [35, +\infty)
 \end{array}
\]
 To decide the minimum order of reduce model satisfying the error tolerance, the FF-type error bound provided by PFDBT and the EF-type error bound provided by LyaBT are plotted in Fig. 5.
 \begin{figure}[ht!]
    \centering
      \includegraphics[scale=0.58]{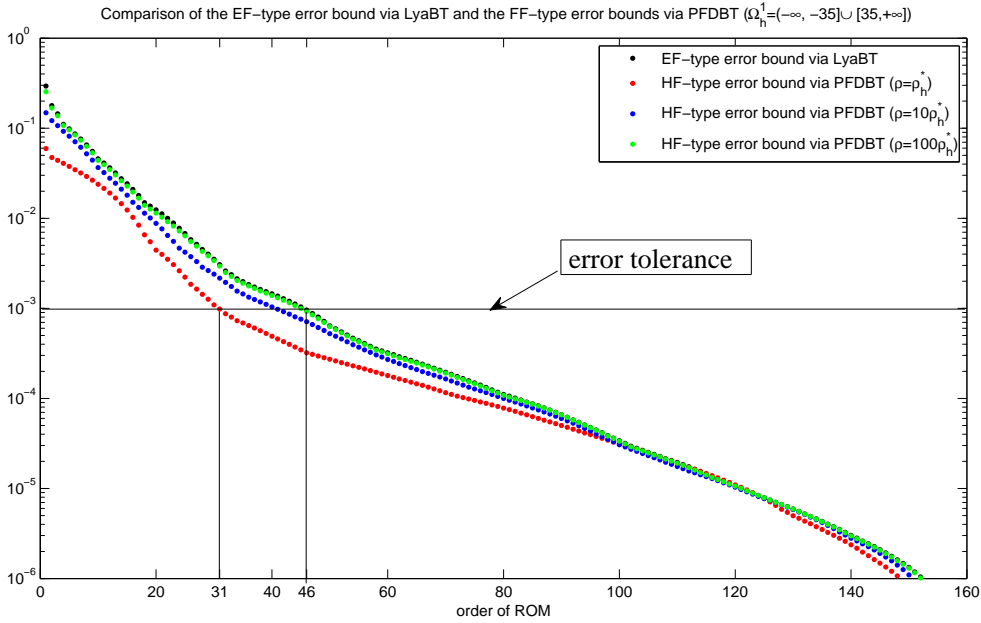}
               \caption{Deciding the minimum order of reduced model by using the error bounds}
       \end{figure}

\noindent As shown by Fig. 5, choosing the $31^{th}$ reduced order model is enough if we adopt the PFDBT. In contrast, $46^{th}$ reduced order model is required if we use the standard LyaBT.

     \begin{figure}[ht!]
    \centering
      \includegraphics[scale=0.58]{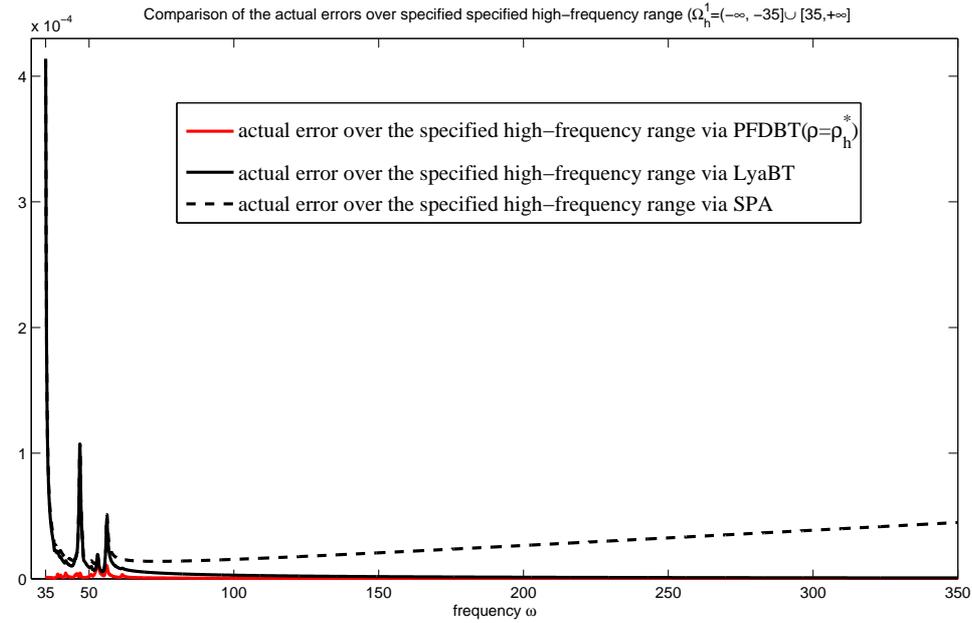}
                 \caption{Actual in-band approximation errors obtained by LyaBT, SPA, and the proposed PFDBT}
       \end{figure}
\noindent Fig. 6 illustrates the actual in-band approximation errors between the original model and the $31^{th}$ reduced models obtained via LyaBT, SPA and the proposed PFDBT,  where $\rho_h^{*}, i=1,2,3,4$ is the minimum value rendering the PFD mapped system $\hat G_{h \rho c} (e^{\jmath \theta})$ Schur stable.  Obviously, PFDBT yields the best in-band approximation performance. Besides, it is shown that both the EF-type error bound and the FF-type error bound are not tight. In fact, all the $31^{th}$ reduced models satisfy the in-band error tolerance. However, only the $31^{th}$ reduced model generated via PFDBT is pre-known to satisfy the in-band error tolerance.
\end{example}

\section{Conclusions and Future Work}
In this paper, we have proposed new parameterized frequency-dependent balanced truncation (PFDBT) schemes to solve some finite frequency (FF) MOR problems. Specifically, the merit of our approach is a family of PFD mapped systems of a given LTI system in the presence of a specified frequency range. We have shown that the finite-frequency maximum singular values of the given system can be bounded by the entire-frequency maximum singular value of the PFD mapped systems. Furthermore, PFDBT schemes solving the LF-MOR (lower frequency) and HF-MOR (higher frequency) problems while providing LF-type and HF-type error bounds are derived by utilizing the PFD bounded real lemmas.  Numerical examples illustrate the results with a comparison between the proposed approach and the standard BT and SPA methods. As future work, it would be interesting to study the MF-MOR (middle frequency) problem in a similar way, i.e., to develop a MF-case PFDBT scheme generating real reduced-order models while providing an MF-type error bound.

\bibliographystyle{IEEEtran}

\end{document}